\documentclass[dvipsnames]{article}

\usepackage{amsmath,amssymb,amsthm,color,bm,mathrsfs,extarrows,tikz,graphicx,mathtools,enumitem}
\usepackage[noadjust]{cite}
\usepackage[ruled]{algorithm2e}
\usepackage[margin=1.in]{geometry}
\usepackage{xcolor,bbm}

\usepackage[pagebackref]{hyperref}
\usepackage{comment}
\hypersetup{
colorlinks=true,
urlcolor=blue,
linkcolor=blue,
citecolor=[rgb]{.42,.56,.14},
}

\usepackage[capitalise,nameinlink]{cleveref}
\Crefname{lemma}{Lemma}{Lemmas}
\Crefname{fact}{Fact}{Facts}
\Crefname{theorem}{Theorem}{Theorems}
\Crefname{corollary}{Corollary}{Corollaries}
\Crefname{claim}{Claim}{Claims}
\Crefname{example}{Example}{Examples}
\Crefname{problem}{Problem}{Problems}
\Crefname{definition}{Definition}{Definitions}
\Crefname{assumption}{Assumption}{Assumptions}
\Crefname{subsection}{Subsection}{Subsections}
\Crefname{section}{Section}{Sections}

\usetikzlibrary{decorations.pathreplacing,decorations.markings,decorations.pathmorphing,decorations.shapes,arrows.meta,positioning}

\newtheorem{theorem}{Theorem}[section]
\newtheorem*{theorem*}{Theorem}

\newtheorem{proposition}[theorem]{Proposition}
\newtheorem*{proposition*}{Proposition}
\newtheorem{lemma}[theorem]{Lemma}
\newtheorem*{lemma*}{Lemma}
\newtheorem{corollary}[theorem]{Corollary}
\newtheorem*{conjecture*}{Conjecture}
\newtheorem{fact}[theorem]{Fact}
\newtheorem*{fact*}{Fact}

\newtheorem*{exercise*}{Exercise}

\newtheorem*{hypothesis*}{Hypothesis}

\theoremstyle{definition}
\newtheorem{definition}[theorem]{Definition}
\newtheorem{construction}[theorem]{Construction}

\newtheorem{exercise-easy}[theorem]{Exercise}
\newtheorem{exercise-med}[theorem]{Exercise}
\newtheorem{exercise-hard}[theorem]{Exercise$^\star$}
\newtheorem{claim}[theorem]{Claim}
\newtheorem*{claim*}{Claim}

\newtheorem{remark}[theorem]{Remark}
\newtheorem*{remark*}{Remark}

\newtheorem*{observation*}{Observation}

\crefformat{equation}{(#2#1#3)}

\usepackage{prettyref}
\newcommand{\savehyperref}[2]{\texorpdfstring{\hyperref[#1]{#2}}{#2}}

\newrefformat{eq}{\savehyperref{#1}{\textup{(\ref*{#1})}}}
\newrefformat{ineq}{\savehyperref{#1}{\textup{(\ref*{#1})}}}
\newrefformat{lem}{\savehyperref{#1}{Lemma~\ref*{#1}}}
\newrefformat{def}{\savehyperref{#1}{Definition~\ref*{#1}}}
\newrefformat{thm}{\savehyperref{#1}{Theorem~\ref*{#1}}}
\newrefformat{cor}{\savehyperref{#1}{Corollary~\ref*{#1}}}
\newrefformat{cha}{\savehyperref{#1}{Chapter~\ref*{#1}}}
\newrefformat{sec}{\savehyperref{#1}{Section~\ref*{#1}}}
\newrefformat{app}{\savehyperref{#1}{Appendix~\ref*{#1}}}
\newrefformat{tab}{\savehyperref{#1}{Table~\ref*{#1}}}
\newrefformat{fig}{\savehyperref{#1}{Figure~\ref*{#1}}}
\newrefformat{hyp}{\savehyperref{#1}{Hypothesis~\ref*{#1}}}
\newrefformat{alg}{\savehyperref{#1}{Algorithm~\ref*{#1}}}
\newrefformat{rem}{\savehyperref{#1}{Remark~\ref*{#1}}}
\newrefformat{item}{\savehyperref{#1}{Item~\ref*{#1}}}
\newrefformat{step}{\savehyperref{#1}{step~\ref*{#1}}}
\newrefformat{conj}{\savehyperref{#1}{Conjecture~\ref*{#1}}}
\newrefformat{fact}{\savehyperref{#1}{Fact~\ref*{#1}}}
\newrefformat{fac}{\savehyperref{#1}{Fact~\ref*{#1}}}
\newrefformat{prop}{\savehyperref{#1}{Proposition~\ref*{#1}}}
\newrefformat{prob}{\savehyperref{#1}{Problem~\ref*{#1}}}
\newrefformat{claim}{\savehyperref{#1}{Claim~\ref*{#1}}}
\newrefformat{relax}{\savehyperref{#1}{Relaxation~\ref*{#1}}}
\newrefformat{rem}{\savehyperref{#1}{Remark~\ref*{#1}}}
\newrefformat{red}{\savehyperref{#1}{Reduction~\ref*{#1}}}
\newrefformat{part}{\savehyperref{#1}{Part~\ref*{#1}}}
\newrefformat{ex}{\savehyperref{#1}{Exercise~\ref*{#1}}}
\newrefformat{property}{\savehyperref{#1}{Property~\ref*{#1}}}
\newrefformat{case}{\savehyperref{#1}{Case~\ref*{#1}}}
\newrefformat{cat}{\savehyperref{#1}{Category~\ref*{#1}}}
\newrefformat{obs}{\savehyperref{#1}{Observation~\ref*{#1}}}
\newrefformat{cond}{\savehyperref{#1}{Condition~\ref*{#1}}}
\newrefformat{que}{\savehyperref{#1}{Question~\ref*{#1}}}

\DeclareSymbolFont{largesymbolsyhmath}{OMX}{yhex}{m}{n}
\DeclareMathAccent{\widehat}{\mathord}{largesymbolsyhmath}{"62}
\DeclareMathAccent{\widetilde}{\mathord}{largesymbolsyhmath}{"65}
\DeclareMathOperator*{\E}{\mathbb E}

\renewcommand{\Pr}{\operatorname*{\mathbf{Pr}}}

\newcommand{\abs}[1]{\left| #1 \right|}

\newcommand{\pbra}[1]{\left( #1 \right)}
\newcommand{\sbra}[1]{\left[ #1 \right]}
\newcommand{\cbra}[1]{\left\{ #1 \right\}}
\newcommand{\bin}{\{0,1\}}
\newcommand{\poly}{\mathrm{poly}}
\newcommand{\ed}{\mathsf{ed}}
\newcommand{\Ham}{\mathsf{Ham}}
\newcommand{\LCS}{\mathsf{LCS}}
\newcommand{\sing}{\mathsf{sing}}
\newcommand{\pred}{\mathsf{pred}}
\newcommand{\sketchx}{{sx}}
\newcommand{\sketchy}{{sy}}

\newcommand{\Qscr}{\mathscr{Q}}

\newcommand{\Nbb}{\mathbb{N}}

\newcommand{\Zbb}{\mathbb{Z}}

\newcommand{\Acal}{\mathcal{A}}

\newcommand{\Ecal}{\mathcal{E}}

\newcommand{\Mcal}{\mathcal{M}}

\newcommand{\Pcal}{\mathcal{P}}
\newcommand{\Scal}{\mathcal{S}}
\newcommand{\Tcal}{\mathcal{T}}

\newcommand{\Zcal}{\mathcal{Z}}

\title{An Improved Sketching Algorithm for Edit Distance}
\author{
Ce Jin\thanks{MIT. \texttt{cejin@mit.edu}. Supported by an Akamai Presidential Fellowship.}
\and
Jelani Nelson\thanks{UC Berkeley. \texttt{minilek@berkeley.edu}. Supported by NSF
    award CCF-1951384, ONR grant N00014-18-1-2562, ONR DORECG award N00014-17-1-2127, and a Google Faculty Research Award.}
\and
Kewen Wu\thanks{UC Berkeley. \texttt{shlw\_kevin@hotmail.com}.}
}
\date{}

\begin{document}
\maketitle

\begin{abstract}
  We provide improved upper bounds for the simultaneous sketching complexity of edit distance. 
Consider two parties, Alice with input $x\in\Sigma^n$ and Bob with input $y\in\Sigma^n$, that share public randomness and are given a promise that the edit distance $\ed(x,y)$ between their two strings is at most some given value $k$. 
  Alice must send a message $\sketchx$ and Bob must send $\sketchy$ to a third party Charlie, who does not know the inputs but shares the same public randomness and also knows $k$. Charlie must output $\ed(x,y)$ precisely as well as a sequence of $\ed(x,y)$ edits required to transform $x$ into $y$. 
  The goal is to minimize the lengths $|\sketchx|, |\sketchy|$ of the messages sent. 
  
  The protocol of Belazzougui and Zhang (FOCS 2016), building upon the random walk method of Chakraborty, Goldenberg, and Kouck\'y (STOC 2016), achieves a maximum message length of $\tilde O(k^8)$ bits, where $\tilde O(\cdot)$ hides $\poly(\log n)$ factors. In this work we build upon Belazzougui and Zhang's protocol and provide an improved analysis demonstrating that a slight modification of their construction achieves a bound of $\tilde O(k^3)$.
\end{abstract}

\section{Introduction}\label{sec:introduction}
The edit distance $\ed(x,y)$ between two strings is defined to be the minimum number of character insertions, deletions, or substitutions required to transform $x$ into $y$. It is one of the most well-studied distance measures on strings, with applications in information retrieval, natural language processing, and bioinformatics. If $x,y$ are each at most length $n$, the textbook Wagner-Fischer algorithm computes $\ed(x,y)$ exactly in $O(n^2)$ time, with the only improvement since being  by a $\log n$ factor due to Masek and Paterson \cite{MasekP80}. It has since been shown that an $O(n^{2-\epsilon})$ time algorithm does not exist for any constant $\epsilon>0$ unless  the \emph{Strong Exponential Time Hypothesis} fails \cite{BackursI18}.
Since the work of \cite{MasekP80}, several subsequent works have considered setups beyond offline exact algorithms for edit distance, such as faster approximation algorithms \cite{AndoniO12,AndoniKO10,ChakrabortyDGKS18,BrakensiekR20,KouckyS20,AndoniN20}, metric embeddings \cite{OstrovskyR07,CharikarK06,KhotN06,KrauthgamerR09}, smoothed complexity \cite{AndoniK12, BoroujeniSS20}, quantum algorithms \cite{BoroujeniEGHS18}, sublinear time algorithms for gap versions \cite{Bar-YossefJKK04,GoldenbergKS19,BrakensiekCR20,KociumakaS20}, and communication complexity and sketching/streaming \cite{ChakrabortyGK16,BelazzouguiZ16,ChengJLW18,Haeupler19,ChengL20}. In this work we focus on communication complexity, and specifically {\it simultaneous communication complexity}.

In the communication model, Alice has input string $x$ and Bob has $y$. They, or a third party, would like to compute $\ed(x,y)$ as well as a minimum length sequence of edits for transforming $x$ into $y$. We consider the setting of shared public randomness amongst all parties. The one-way setting in which Alice sends a single message to Bob, who must then output $\ed(x,y)$, is known as the {\it document exchange problem} and has a long history. In the promise version of the problem for which we are promised $\ed(x,y) \le k$, Orlitsky \cite{Orlitsky91} gave a deterministic protocol in which Alice only sends $O(k\log(n/k))$ bits in the case of binary strings, which is optimal, with the downside that Bob's running time to process her message is exponential. Haeupler recently used public randomness to improve Bob's running time to polynomial with the same asymptotic message length, and it is now known that a polynomial-time recovery algorithm is achievable deterministically if one increases the message length to $O(k\log^2(n/k))$ \cite{ChengJLW18,Haeupler19}. Belazzougui and Zhang \cite{BelazzouguiZ16} studied the harder {\it simultaneous communication} model in which Alice and Bob each send messages to a third party Charlie, who knows neither string but shares knowledge of the public randomness, and Charlie must output $\ed(x,y)$ as well as the edits required to transform $x$ into $y$. In this model they gave a protocol in which each player sends $O(k^8 \log^5 n) = \tilde O(k^8)$ bits.\footnote{We use $\tilde O(f)$ throughout this paper to denote $f\cdot \mathop{\mathrm{polylog}}(n)$.}

\begin{definition}[Problem $\Qscr_{n,k,\delta}$]\label{def:problem}
Alice and Bob and a referee share public randomness. 
Alice (resp., Bob) gets a length-$n$ input string $x$ (resp., $y$) over alphabet $\Sigma$, and then sends a ``sketch'' $\sketchx\in\{0,1\}^*$ (resp., $\sketchy$) to the referee. We say the {\it size} of the sketch is maximum length of strings $\sketchx$ and $\sketchy$.
After receiving the sketches $sx$ and $sy$, 
\begin{itemize}
\item if $\ed(x,y)\le k$, the referee needs to compute $\ed(x,y)$ as well as an optimal edit sequence from $x$ to $y$, with success probability at least $1-\delta$;
\item if $\ed(x,y)>k$, the referee needs to report ``error'', with success probability at least $1-\delta$.
\end{itemize}
\end{definition}

\paragraph{Main contribution.} We build upon and improve techniques developed in \cite{BelazzouguiZ16} to show that a very slight modification of their protocol needs a sketch size of only $\tilde O(k^3)$ bits to solve problem $\Qscr_{n,k,\delta}$. More precisely, the bound is $O(k^3\log^2(n/\delta)\log n)$ bits.\footnote{We remark that both the algorithm of \cite{BelazzouguiZ16} and our improved algorithm are time-efficient, and work in the more restrictive setting where Alice and Bob have only $\poly(k\log (n/\delta))$ memory and receive the input strings in a \emph{streaming} fashion. }


\subsection{Proof Overview}\label{sec:proof_overview}

We provide a high-level description of the previous results \cite{ChakrabortyGK16,BelazzouguiZ16} that we build on, and then briefly describe our new ideas.

\paragraph{CGK random walk.}
The previous sketching result \cite{BelazzouguiZ16} uses a random walk technique developed in \cite{ChakrabortyGK16}. Given two input strings $x,y$ of length $n$, we append them with infinitely many zeros and initialize two pointers $i=1,j=1$. 
In each step $t$, we first append $x[i]$ to Alice's output tape (and append $y[j]$ to Bob's output tape), and then increment $i$ by $r_t(x[i]) \in \{0,1\}$, and increment $j$ by $r_t(y[j])\in \{0,1\}$, where $r_t\colon \Sigma \to \{0,1\}$ is a random function. 
The process continues for $3n$ steps and we consider the evolution of $i-j$, i.e., the distance between the two pointers during this random process. 
Observe that when $x[i]\neq y[j]$, the change of $i-j$ is a mean-zero random variable in $\{-1,0,+1\}$ (and we call this a \emph{progress step}); while when $x[i]=y[j]$, the difference $i-j$ will not change. 

The main result of \cite{ChakrabortyGK16} shows that the number of progress steps in this random process is at least $\ed(x,y)/2$, and at most $O(\ed(x,y)^2)$ with constant probability. 
This property was used to design a sketching protocol (with public randomness for generating $r_t$) for estimating $\ed(x,y)$ up to a quadratic factor error by applying an approximate Hamming distance sketching protocol to the two strings generated by the random walk (where a progress step corresponds to a Hamming mismatch between Alice's and Bob's output strings).

\paragraph{\cite{BelazzouguiZ16} algorithm.}
The key idea of \cite{BelazzouguiZ16} is the following. A CGK random walk naturally induces a non-intersecting matching between the input strings: we view $x$ and $y$ as a bipartite graph, where if $(i,j)$ is an edge then $x[i]=y[j]$ and $i,j$ are the pointers in some step of the walk.
In particular, this matching can be viewed as an edit sequence where a character is unchanged if it is matched.

Using an exact Hamming sketch protocol (with sketch size near-linear in the number of Hamming errors), the referee can recover this matching, as well as all the unmatched characters. Although this matching may not correspond to an edit sequence of optimal length, \cite{BelazzouguiZ16} shows: suppose we obtain \emph{multiple} matchings by running i.i.d.~CGK random walks. Then, 
\begin{enumerate}[label=(\alph*)]
    \item if the \emph{intersection} of their matched edges is contained in an optimal matching, then one can extract enough information from the matchings and unmatched characters to recover an optimal edit sequence using dynamic programming;
    \item if we generate $\poly(\ed(x,y),\log n)$ many i.i.d.~CGK random walks, then the precondition of Item (a) is satisfied with constant probability.
\end{enumerate} 

\paragraph{Our improvements.} 
We obtain our result by improving the dependence on $\ed(x,y)$ in Item (b) described above. In particular, we reduce the number of required random walks. Our improvements come from two parts.

To obtain the first improvement, we observe that \cite{BelazzouguiZ16}'s algorithm relies on the following two events happening. The first is that, for every edge that does appear in  a (fixed) optimal matching, there should be one of the sampled CGK random walks that misses this edge. The second is that the CGK random walks should have few progress steps. 
In \cite{BelazzouguiZ16}, they pay a union bound over the two events to make sure \emph{all CGK random walks} are good for the decoder. This introduces a large dependence on $\ed(x,y)$, mainly due to the fact that the number of  Hamming errors in a CGK random walk has a heavy-tailed distribution.
We manage to avoid this by arguing that these two events happen simultaneously (see \Cref{lem:key_lemma}) with decent probability, and then modifying the decoding algorithm to only consider \emph{those good CGK random walks}.

The second improvement comes from improved analysis for Item (b), which depends on the following property of the CGK random walk \cite[Lemma 16]{BelazzouguiZ16} (see \Cref{sec:proof_of_lem:key_lemma_small_gap} for how this property can be used): informally, if a string $X[1..L]$ has a certain kind of self-similarity (for example, it is periodic), then with some nontrivial probability, a CGK random walk on $X$ itself starting with two pointers $i=2,j=1$ will not pass through the state $(i=L,j=L)$. To be more precise, if there is a non-intersecting matching between $X[1..L]$ and itself, where every matched edge $(I,J)$ satisfies $I>J$, and the number of singletons (unmatched characters) is at most $K$, then the CGK random walk will miss $(i=L,j=L)$ with $\Omega(1/K^2)$ probability.\footnote{There is a subtle gap in the proof of \cite[Lemma 16]{BelazzouguiZ16}. On page 18 of their full version, they bounded the number of progress steps in two cases: (1) at least one of the pointers is not in any cluster, and (2) both of the two pointers are in the same cluster. (Their terminology \emph{cluster} refers to a contiguous sequence of matched edges with no singletons in-between.) However, they did not analyze the case where the two pointers are separated in different clusters, and it was not clear to us how to repair that gap using the techniques developed in \cite{BelazzouguiZ16}.}

We use a more technical analysis to improve the bound to $\Omega(1/K)$ (see \Cref{prop:bounds_on_rho}). Now we informally describe our main idea. Starting from the state $(i=2,j=1)$, with at least $\Omega(1/K)$ probability it will first reach a state $(i,j)$ with $i-j>d_0=\Theta(K)$ before reaching $i-j=0$ (note that $i-j$ can never become negative). Then we will show that with good probability $i-j$ will remain in the range $[d_0/2, 3d_0/2]$. To do this, we show an $O(d_0^2)$ upper bound on the expected total number of progress steps, and use the fact that the expected deviation produced by a $P$-step one-dimensional random walk is $O(\sqrt{P})$. 

To bound the expected total number of progress steps, we divide the evolution of the state $(i,j)$ into several phases, where in each phase the pointers move from a \emph{stable state} $(i,j)$ to another \emph{stable state} $(i',j')$, satisfying $j'\ge i$ and $i'\ge 2j-i$. Here, a \emph{stable state} $(i,j)$ informally means that we have a good upper bound of $\ed(X[j..i-1],X[i..2i-j-1])$ in terms of the number of singletons in the range $[j..2i-j-1]$ (for example, if $X$ is ``close'' to a string with period $p$, and $i-j$ is approximately a multiple of $p$, then $(i,j)$ is a stable state).
We will bound the expected number of progress steps in one phase by $O\pbra{(i-j)\cdot S+ S^2}$, where $S$ denotes the number of singletons in the range $[j..i'-1]$. We can see the sum of $S$ over all phases is at most $2K$ since each singleton is counted at most twice. Hence, summing up over all phases would give the desired $O(K^2)$ upper bound, if we assume $i-j = \Theta(d_0)$.
Although this assumption may lead to circular reasoning, we can get around this issue by a more careful argument.

\paragraph{Organization.} We give several needed definitions in \Cref{sec:preliminary}. In \Cref{sec:sketches_for_edit_distance} we state and analyze our sketching algorithm, which as mentioned, is mostly similar to \cite{BelazzouguiZ16} but with small modifications. \Cref{sec:bounds_on_rho} is devoted to our main technical lemma. In comparison with the proof overview, \Cref{sec:sketches_for_edit_distance} is for the first improvement and \Cref{sec:bounds_on_rho} is for the second improvement. Then we discuss limits on our approach and further problems in \Cref{sec:discussions}. The lower bounds and some of the technical proofs are deferred to the appendix.
\section{Preliminaries}\label{sec:preliminary}

In this section we introduce formal definitions.

\subsection{Notations}
Let $[n]$ denote $\{1,2,\dots,n\}$, and let $[l..r]$ denote $\{l,l+1,\dots,r\}$. Let $\circ$ denote string concatenation.
Let $\Nbb$ denote the set of natural numbers $\cbra{0,1,\ldots}$.
We consider sketching protocols for strings in $\Sigma^n$ in this work, where $\Sigma$ denotes the alphabet. We assume $|\Sigma| \le \poly(n)$ and $0\in\Sigma$.\footnote{For larger alphabet the algorithm still works but some $\log n$ terms in the bounds become $\log|\Sigma|$. For example, the sketch size will be $O\pbra{k^3\log(n|\Sigma|/\delta)\log(n/\delta)\log n}$. Alternatively, the parties can hash $\Sigma$ into a new alphabet of size $O(n^2/\gamma)$ and have no hash collisions on the characters appearing in $x,y$ with probability at least $1-\gamma$.}

For a string $s \in \Sigma^{n}$ and index $1\le i \le n$, $s[i]$ (or sometimes $s_i$) denotes the $i$-th character of $s$. For $1\le i \le j\le n$, $s[i..j]$ denotes the substring $s[i]\circ s[i+1]\circ\cdots\circ s[j]$. If $i>j$ then $s[i..j]$ is the empty string.

\subsection{Edit Distance}
\begin{definition}[Edit distance $\ed(\cdot,\cdot)$]
The \emph{edit distance} between two strings $x$ and $y$, denoted by $\ed(x,y)$, is the minimum number of edits (insertions, deletions, and substitutions\footnote{There is another definition of edit distance, denoted by $\ed'(x,y)$, where only insertions and deletions are allowed. 
We have $\ed(x,y) \le \ed'(x,y)\le 2\cdot \ed(x,y)$, and $\ed'(x,y) = |x|+|y|-\LCS(x,y)$, where $\LCS$ stands for \emph{longest common subsequence}.
The algorithm in \cite{BelazzouguiZ16}, as well as our modification of it, can be easily adapted to work for this variant of edit distance as well.
}) required to transform $x$ to $y$.
\end{definition}

We note the following simple facts about edit distance. 
\begin{fact}\label{fct:length_to_edit}
Let $x$ and $y$ be two strings of length $n$ and $m$ respectively. Then $\ed(x,y)\ge|n-m|$.
\end{fact}

\begin{proposition}\label{prop:edit_dist}
Let $x,y$ be two length-$n$ strings. Let $x'$ be any (not necessarily contiguous) subsequence of $x$. Then $\ed(x,y) \le 2\cdot \ed(x',y)$.
\end{proposition}
\begin{proof}
Since $\ed(x',y)\ge n-\LCS(x',y)\ge n-\LCS(x,y)$, we have $\ed(x,y)\le 2\cdot(n-\LCS(x,y))\le2\cdot\ed(x',y)$.
\end{proof}

\begin{definition}[Matching induced by edit sequence $\Mcal(S)$]
Given strings $x,y$ and an edit sequence $S$, we can construct a bipartite graph between $x$ and $y$, where every character in $x$ that is not substituted nor deleted is connected by an edge to its counterpart in $y$.
These edges form a non-intersecting matching, which we denote by $\Mcal(S)$.
Moreover, when $S$ achieves optimal edit distance, we say $\Mcal(S)$ is an \emph{optimal matching}.
\end{definition}

We show the following properties of an optimal matching, the proof of which is deferred to \Cref{app:optimal_matching}.

\begin{lemma}\label{lem:optimal_matching}
Let $x,y$ be two strings.
Let $S$ be an optimal edit sequence and $\Mcal(S)$ be its corresponding optimal matching.
\begin{enumerate}[label=(\arabic*)]
\item If $(i,j)\in\Mcal(S)$, then $|i-j|\le\ed(x,y)$.
\item If $u'\le u$ and $v'\le v$ and $u-u'+1=v-v'+1=:L$, then the number of matched edges with both endpoints in $x[u'..u]$ and $y[v'..v]$ is at least $L-3\cdot\ed(x,y)-|u-v|$, i.e.,
$$
\abs{\Mcal(S)\cap\big([u'..u]\times[v'..v]\big)}\ge L-3\cdot\ed(x,y)-|u-v|.
$$
\end{enumerate}
\end{lemma}

Though there may be multiple optimal matchings, the following definition specifies a canonical one.

\begin{definition}[Greedy optimal matching $\Mcal$, \cite{BelazzouguiZ16}]
Let $x,y$ be two strings. For each edit sequence $S$ achieving optimal edit distance, let $\Mcal(S)$ be the matching induced by $S$.
Then the \emph{greedy optimal matching $\Mcal$} is defined to be the smallest $\Mcal(S)$ in lexicographical order. Specifically, we represent $\Mcal(S)$ as a sequence of $(i,j)$ pairs then sort the sequence lexicographically, and the greedy optimal matching is such that this sorted sequence is as lexicographically small as possible.
\end{definition}

This greedy optimal matching enjoys some extra properties, which can be easily proved.

\begin{lemma}[\cite{BelazzouguiZ16}]\label{lem:greedy_optimal_matching}
Let $x,y$ be two strings and $\Mcal$ be their greedy optimal matching. 
\begin{enumerate}[label=(\arabic*)]
\item If $(i,j)\in\Mcal$ and $x[i+1]=y[j+1]$, then $(i+1,j+1)\in\Mcal$.
\item If $x[u'..u]=y[v'..v]$ and $(i,j),(i',j')\in\Mcal\cap \big([u'..u]\times[v'..v]\big)$ are two matched edges in $x[u'..u],y[v'..v]$, then $((i-j)-(u-v))\cdot((i'-j')-(u-v))\ge0$. Moreover, when the equality holds we have $(u,v)\in\Mcal$.
\end{enumerate}
\end{lemma}

\subsection{The CGK Random Walk}
We review a useful random process called the \emph{CGK random walk}, which was first introduced by Chakraborty, Goldenberg, and Kouck{\'{y}} \cite{ChakrabortyGK16}, and played a central role in the sketching algorithm of \cite{BelazzouguiZ16}.

\begin{definition}[CGK random walk $\lambda_r(s)$, \cite{ChakrabortyGK16}]
\label{def:cgk}
Given a string $s\in \Sigma^n$, an integer $m\ge 0$, and a sequence of $m\cdot |\Sigma|$ random coins interpreted as a random function $r\colon [m] \times \Sigma \to \{0,1\}$, the $m$-step CGK random walk is a length-$m$ string $\lambda_r(s) \in \Sigma^m$ defined by the following process:
\begin{itemize}
\item Append $s$ with infinitely many zeros.
\item Initialize the pointer $p\gets1$ and the output string $s'\gets\emptyset$.
\item For each step $i=1,\ldots,m$:
\begin{itemize}
\item Append $s[p]$ to $s'$.
\item Update $p\gets p+r(i,s[p])$.
\end{itemize}
\item Output $s'=: \lambda_r(s)$.
\end{itemize}
For a contiguous segment of the output string $\lambda_r(s)$, the \emph{pre-image} of this segment refers to the corresponding substring in the original input string $s$ (which may also include the appended trailing zeros if the walk extends beyond $s$).

Due to its usefulness in the two-party setting with public randomness, we also frequently use the term \emph{CGK random walk} to refer to a \emph{pair} of random walks (as defined in \Cref{def:cgk}) performed on two input strings $x,y$ using the \emph{shared} random string $r$.

Consider a CGK random walk $\lambda$ on two input strings $x,y$. 
We use $p_i$ (resp., $q_i$) to denote the pointer on string $x$  (resp., $y$) at the beginning of step $i$. We refer to the pair $(p_i,q_i)$ as the \emph{state} of $\lambda$ at the $i$-th step, and  we write $(p,q)\in \lambda$ if $\lambda$ passes through the state $(p,q)$, i.e., there exists some $i$ for which $p_i=p$ and $q_i=q$.
We say the $i$-th step of $\lambda$ is a \emph{progress step} if the $i$-th characters of the output strings $\lambda_r(x)$ and $\lambda_r(y)$ differ, or equivalently, $x[p_i]\neq y[q_i]$.\footnote{Our definition of ``progress step'' is different from that of \cite{BelazzouguiZ16}, which additionally requires at least one of the two pointers moves forward in that step.}
We say $\lambda$ \emph{walks through} $x,y$, if in the end the  two pointers satisfy $p_m\ge |x|$ and $q_m\ge |y|$.
\end{definition}

The following theorem established the connection between CGK random walks and edit distance. Informally, when $\ed(x,y)$ is small, with good probability the number of progress steps in $\lambda$ is also small (or equivalently, the Hamming distance between the output strings $\lambda_r(x), \lambda_r(y)$ is small).
\begin{theorem}[{\cite[Theorem 4.1]{ChakrabortyGK16}}]\label{thm:CGK}
Let $\lambda$ be an $m$-step  CGK random walk on $x,y$. Then
\begin{enumerate}[label=(\arabic*)]
\item if $m\ge3\cdot\max\cbra{|x|,|y|}$, then $\lambda$ walks through $x,y$ with probability at least $1-e^{\Omega(m)}$;
\item given $\lambda_r(x)$ and $r$, we can reconstruct the pre-image
of $\lambda_r(x)$ ;
\item $\Pr\sbra{\#\text{progress steps in $\lambda$}\ge \pbra{T\cdot\ed(x,y)}^2}\le O(1/T)$.
\end{enumerate}
\end{theorem}
We provide a simpler proof for Item (3) of this theorem in \Cref{app:simpler_analysis_of_CGK}.

\subsection{Random Walks}

We frequently relate the CGK random walk to the following one-dimensional random walk.

\begin{definition}[One-dimensional unbiased and self-looped random walk]\label{def:1D_random_walk}
A stochastic process $X=(X_t)_{t\in\Nbb}$ on integers is a one-dimensional unbiased and self-looped random walk if its transition satisfies
$$
X_i=\begin{cases}
X_{i-1}-1 & \text{w.p., $1/4$},\\
X_{i-1} & \text{w.p., $1/2$},\\
X_{i-1}+1 & \text{w.p., $1/4$}.
\end{cases}
$$
\end{definition}

\begin{remark}\label{rmk:progress_step_and_random_walk}
Let $\lambda$ be a CGK random walk on two strings and $(p,q)$ be its state. Define $\Delta=p-q$. Then $\Delta$ can be viewed as a one-dimensional unbiased and self-looped random walk, which makes a transition when and only when $\lambda$ makes a progress step. 
\end{remark}

\begin{fact}[e.g.\ {\cite[Proposition 2.1]{levinmarkov}}]\label{thm:steps_in_random_walk}
Let $a,b$ be two non-negative integers and $X$ be a one-dimensional unbiased and self-looped random walk.
Suppose the walk starts at $X_0=0$ and stops when $(X_i=-a)\lor(X_i=b)$. Then 
\begin{enumerate}[label=(\arabic*)]
\item if $(a,b)\neq(0,0)$, then $\Pr\sbra{X\text{ stops at }b}=a/(a+b)$ and $\Pr\sbra{X\text{ stops at }a}=b/(a+b)$;
\item $\E\sbra{\#\text{steps until }X\text{ stops}}=2\cdot ab$.
\end{enumerate}
\end{fact}

By \Cref{rmk:progress_step_and_random_walk} and the martingale property, we have the following lemma, the proof of which is deferred to \Cref{app:CGK_random_walk}.

\begin{lemma}\label{lem:CGK_random_walk}
Consider an $\infty$-step CGK random walk $\lambda$ on $x,y$, where $p,q$ are the pointers on $x,y$ respectively. Let $u$ be an index and let $U,V\ge u-1$ be any integers. Then the following hold.
\begin{enumerate}[label=(\arabic*)]
\item Let $T_0$ be the first time that $p_{T_0}\ge u$. Then $\E\sbra{\abs{p_{T_0}-q_{T_0}}}\le 4\cdot\ed(x[1..U],y[1..V])$.
\item Let $T_1$ be the first time that $(p_{T_1}\ge u)\land(q_{T_1}\ge u)$. Then $\E\sbra{\abs{p_{T_1}-q_{T_1}}}\le 4\cdot\ed(x[1..U],y[1..V])$.
\end{enumerate}
\end{lemma}
\section{Sketches for Edit Distance}\label{sec:sketches_for_edit_distance}

For the rest of the paper, we use the following notational conventions: 
\begin{itemize}
\item $n$ is the length of the input strings; $m:=3n$ is the number of steps in a CGK random walk.
\item $x,y$ are the input strings of length $n$, which is appended with infinitely many zeros; we are promised $\ed(x,y)\le k$.\footnote{We will also analyze the behaviour of our algorithms when $\ed(x,y)>k$.}
\item when we use $(\cdot,\cdot)$ to denote a CGK state or an edge between $x,y$, the first coordinate is a pointer on $x$ and the second is on $y$.
\item $\Mcal$ is the greedy optimal matching of $x,y$.
\end{itemize}
Our goal is to prove the following theorem.
\begin{theorem}\label{thm:main_theorem}
There exists a sketching algorithm for $\Qscr_{n,k,\delta}$ with sketch size $O\pbra{k^3\log^2(n/\delta)\log n}$ bits. Moreover, the algorithm has the following properties.
\begin{itemize}
\item The encoding algorithm used by Alice (resp., Bob) only assumes one-pass streaming access to the input string $x$ (resp., $y$). The time complexity per character is $\poly(k\log(n/\delta))$, and the space complexity is $O\pbra{k^3\log^2(n/\delta)\log n}$ bits. \footnote{ 
The algorithm may use a large number of shared random bits, which can be reduced using Nisan's generator \cite{Nisan92}. The main cost, as we can see from the proof, comes from the CGK random walk. Hence we refer readers to \cite{ChakrabortyGK16} for more details on reducing randomness for the CGK random walk.
}
\item The decoding algorithm used by the referee has time complexity $\poly(k\log(n/\delta))$.
\end{itemize}
\end{theorem}

In \Cref{sec:general_framework}, we review the general framework of \cite{BelazzouguiZ16}'s sketching protocol, and highlight our key improvement in \Cref{lem:key_lemma}. We will prove this key lemma in \Cref{sec:proof_of_lem:key_lemma_large_gap} and \Cref{sec:proof_of_lem:key_lemma_small_gap}. In \Cref{sec:sketch-construction} we present the detailed construction of sketches.

\subsection{General Framework}\label{sec:general_framework}

We adopt the definition of \emph{effective alignments} from \cite{BelazzouguiZ16}. Intuitively, an effective alignment between two strings $x,y$ contains the information of an edit sequence from $x$ to $y$, but does not contain the information of unchanged characters.

\begin{definition}[Effective alignment $\Acal$, \cite{BelazzouguiZ16}]\label{def:effective_alignment}
For two strings $x,y\in\Sigma^n$, an \emph{effective alignment} $\Acal$ between $x$ and $y$ is a triplet $(G,g_x,g_y)$, where 
\begin{itemize}
\item $G=(V_x,V_y,E)$ is a bipartite matching where nodes $V_x=[n],V_y=[n]$ correspond to indices of characters in $x$ and $y$ respectively, and every matched edge $(i,j)\in E$ satisfies $x[i]=y[j]$.
Moreover,  the matched edges  are non-intersecting, i.e., for every pair of distinct edges $(i,j),(i',j')\in E$, we have $i<i'$ iff $j<j'$.
\item $g_x$ (resp., $g_y$) is a partial function defined on the set of unmatched nodes $U_x\subseteq V_x$ (resp., $U_y\subseteq V_y$). For each $i\in U_x$ (resp., $j\in U_y$), define $g_x(i)=x[i]$ (resp., $g_y(j)=y[j]$). 
\end{itemize}
\end{definition}

\begin{definition}[Effective alignments consistent with a CGK random walk, \cite{BelazzouguiZ16}]
Let $\lambda$ be a  CGK random walk on $x,y$, where $p,q$ are the pointers on $x$ and $y$ respectively. If $\lambda$ walks through $x,y$, then we say an effective alignment $\Acal=(G,g_x,g_y)$ is \emph{consistent with} $\lambda$ if for every matched edge $(p,q)\in G$, we have  $(p,q)\in\lambda$.
\end{definition}

As mentioned in \Cref{sec:proof_overview}, Alice and Bob use public randomness to instantiate $\tau=O(k\log(n/\delta))$ independent CGK random walks $\lambda_1,\ldots,\lambda_\tau$ on $x,y$. 
Then, for each CGK random walk $\lambda_i$, Alice constructs a sketch $\sketchx_i$ based on her part of the random walk $\lambda_i(x)$, and Bob similarly constructs $\sketchy_i$ based on his part of the random walk $\lambda_i(y)$.
The referee receives $\sketchx_i,\sketchy_i$, and tries to extract an effective alignment $\Acal_i$ from the sketches.
Each $\sketchx_i$ (and $\sketchy_i$) has length $O(k^2\log(n/\delta)\log n)$.
The properties of this protocol are summarized as follows.
\begin{construction}[Sketch for each random walk, adapting \cite{BelazzouguiZ16}]\label{constr:sketch_for_each_random_walk}
Let $C\ge1$ be some large constant and $\eta\in(0,1)$.
There exists an efficient sketching algorithm such that the following holds. Let $\lambda$ be an $m$-step CGK random walk on $x$ (and $y$). Then,
\begin{itemize}
\item the sketch size and encoding space are $O\pbra{k^2\log(n/\eta)\log n}$ bits;
\item the encoding time per character and decoding time are both $\poly(k\log(n/\eta))$;
\item for fixed $\lambda,x,y$ the following hold with success probability at least $1-\eta$:
\begin{itemize}
\item the decoder either (a) reports ``error'', or (b) outputs an effective alignment $\Acal$ consistent with $\lambda$;
\item when $\lambda$ walks through $x,y$ and contains at most $C\cdot k^2$ progress steps, (b) occurs.
\end{itemize}
\end{itemize}
\end{construction}
We present a formal proof of \Cref{constr:sketch_for_each_random_walk} in \Cref{sec:sketch-construction}. 

The final sketches are simply $\sketchx=\sketchx_1\circ\cdots\circ\sketchx_\tau$ and $\sketchy=\sketchy_1\circ\cdots\circ\sketchy_\tau$. The referee tries to obtain an effective alignment from every $(\sketchx_i,\sketchy_i)$, and then uses  the following lemma  to compute $\ed(x,y)$ and recover an optimal edit sequence. 

\begin{lemma}[{\cite[Lemma 14 and Lemma 19]{BelazzouguiZ16}}]\label{lem:algorithm_of_referee}
There exists a deterministic algorithm taking $(\sketchx,\sketchy)$ as input such that the following holds. 
\begin{itemize}
\item The running time of the algorithm is $\poly(|\sketchx|+|\sketchy|)=\poly(k\log(n/\delta))$. 
\item Let $\Acal_{i_1},\ldots,\Acal_{i_w}$ be the effective alignments\footnote{Although we can check if $\Acal_{i_j}$ is an effective alignment, we cannot verify (without knowing $\lambda_{i_j}$) if $\Acal_{i_j}$ is an effective alignment \emph{consistent with $\lambda_{i_j}$}. This subtle difference comes from that in \Cref{constr:sketch_for_each_random_walk} we do not give any guarantee outside the $1-\eta$ success probability, where the decoder might provide some effective alignment that is \emph{not} consistent with $\lambda_{i_j}$.} decoded from $(\sketchx_1,\sketchy_1),\ldots,(\sketchx_\tau,\sketchy_\tau)$.

If $w\ge1$ and each $\Acal_{i_j}$ is consistent with $\lambda_{i_j}$, then the algorithm outputs a valid edit sequence. 
If, additionally, $\Mcal$ goes through all edges that are common to $\Acal_{i_1},\ldots,\Acal_{i_w}$, then the edit sequence is optimal.
\end{itemize}
\end{lemma}

Now we state our key lemma.
\begin{lemma}[Key Lemma]\label{lem:key_lemma}
There exist some large constants $C_1,C_2\ge1$ such that the following holds. 
Let $\lambda$ be an $\infty$-step CGK random walk on $x,y$. Then for any fixed $(u,v)\notin\Mcal, x[u]=y[v]$, we have
$$
\Pr\sbra{(u,v)\notin\lambda \bigwedge \#\text{progress steps in $\lambda$}\le C_1\cdot k^2}\ge\frac1{C_2\cdot k}.
$$
\end{lemma}
Here we reiterate that \Cref{lem:key_lemma} summarizes our improvement over the previous work of \cite{BelazzouguiZ16} in two aspects (as mentioned in \Cref{sec:proof_overview}): (1) The previous work only gave a lower bound on $\Pr[(u,v)\notin \lambda]$, while we bound the probability of two events happening simultaneously; (2) The previous work only gave a bound of $\Omega(1/k^2)$, while we give an $\Omega(1/k)$ bound.
The proof of this \Cref{lem:key_lemma} is divided into two parts in \Cref{sec:proof_of_lem:key_lemma_large_gap} and \Cref{sec:proof_of_lem:key_lemma_small_gap}, in which a technical proposition that leads to the improvement in Item (2) will be proved in \Cref{sec:bounds_on_rho}.

Assuming \Cref{lem:key_lemma},  we can prove \Cref{thm:main_theorem}. 
\begin{proof}[Proof of \Cref{thm:main_theorem}]
Let $C_3$ be a large constant.

For the encoding part, we instantiate $\tau=C_2k\cdot C_3\log(n/\delta)=O(k\cdot\log(n/\delta))$
independent $m$-step CGK random walks $\lambda_i,i\in[\tau]$; and construct each $\sketchx_i,\sketchy_i$ using \Cref{constr:sketch_for_each_random_walk} with parameter $C=C_1,\eta=\delta/(2\tau)$.

For the decoding part, we run the decoding procedure in \Cref{constr:sketch_for_each_random_walk} to obtain $\Acal_{i_1},\ldots,\Acal_{i_w}$ for \Cref{lem:algorithm_of_referee}. If $w=0$ or the edit sequence from \Cref{lem:algorithm_of_referee} has more than $k$ edits, we report ``error''; otherwise we output the edit sequence together with the corresponding edit distance.

\paragraph*{\fbox{Bounds on the parameters.}} By constructing each $\sketchx_i$ (and $\sketchy_i$) in parallel, the final sketch size and encoding space\footnote{We omit the space for storing auxiliary information (e.g., pointers) in the calculation, since these are minor terms.} are
$$
\tau\cdot O\pbra{k^2\log(n/\eta)\log n}=O\pbra{k^3\log^2(n/\delta)\log n}.
$$
The encoding time per character is then
$$
\tau\cdot\poly\pbra{k\log(n/\eta)}=\poly(k\log(n/\delta)).
$$
The decoding time follows immediately from \Cref{lem:algorithm_of_referee}.

\paragraph*{\fbox{Analysis of the algorithm when $\ed(x,y)>k$.}}
Since $\eta=\delta/(2\tau)$ in \Cref{constr:sketch_for_each_random_walk} and by union bound, the decoder, with probability at least $1-\delta/2$, for each $(\sketchx_i,\sketchy_i)$ either reports ``error'', or outputs an effective alignment $\Acal_i$ consistent with $\lambda_i$. 
Conditioning on this, when we apply \Cref{lem:algorithm_of_referee}, either $w=0$ or it outputs a valid edit sequence which has more than $k$ edits. In either case, we report ``error''.

\paragraph*{\fbox{Analysis of the algorithm when $\ed(x,y)\le k$.}}
Let $S=\cbra{(u,v)\in[n]^2\mid(u,v)\notin\Mcal,x[u]=y[v]}$ and define events 
\begin{itemize}
\item $\Ecal_i$: $\lambda_i$ walks through $x,y$.
\item $\Ecal_i'(u,v)$ for $(u,v)\in S$: $(u,v)\notin\lambda_i \bigwedge \#\text{progress steps in $\lambda_i$}\le C_1\cdot k^2$.
\end{itemize}
Then
\begin{align}
&\phantom{\ge}\Pr\sbra{\forall(u,v)\in S,~\exists i\in[\tau],~\Ecal_i\land\Ecal_i'(u,v)}\notag\\
&\ge1-\sum_{(u,v)\in S}\Pr\sbra{\Ecal_1\land\Ecal_1'(u,v)}^\tau\notag\\
&\ge1-\sum_{(u,v)\in S}\pbra{1-\Pr\sbra{\neg\Ecal_1}-\Pr\sbra{\neg\Ecal_1'(u,v)}}^\tau\notag\\
&\ge1-n^2\cdot\pbra{1-e^{\Omega(n)}-\frac1{C_2\cdot k}}^\tau
\tag{due to \Cref{thm:CGK} and \Cref{lem:key_lemma}}
\\
&\ge1-\frac\delta2.\label{eq:main_theorem_1} 
\end{align}

Let $\lambda_{i_1},\ldots,\lambda_{i_w}$ be the random walks walking through $x,y$ and containing at most $C_1\cdot k^2$ progress steps. 
Since $\eta=\delta/(2\tau)$ in \Cref{constr:sketch_for_each_random_walk} and by union bound, the decoder, with probability at least $1-\delta/2$, for each $(\sketchx_i,\sketchy_i)$ either reports ``error'', or outputs an effective alignment $\Acal_i$ consistent with $\lambda_i$.
Conditioning on this, \Cref{constr:sketch_for_each_random_walk} must at least obtain effective alignments $\Acal_{i_j},\ldots,\Acal_{i_w}$ that are consistent with the corresponding random walks.
Combined with \Cref{eq:main_theorem_1}, with probability at least $1-\delta$, for any $(u,v)\in S$ there exists some $\lambda_{i_j}$ missing it.
Then the edit sequence from \Cref{lem:algorithm_of_referee} is optimal.
\end{proof}

\subsection[Proof of Lemma 3.6: Case |u-v|>100k]{Proof of \Cref{lem:key_lemma}: Case $|u-v|>100\cdot k$}\label{sec:proof_of_lem:key_lemma_large_gap}

\begin{proof}[Proof of \Cref{lem:key_lemma}: Case $|u-v|>100\cdot k$]
Assume without loss of generality $u>v$. We stop $\lambda$ when it meets $u$. Then by Item (1) in \Cref{lem:CGK_random_walk}, at this time the state $(p,q)$ satisfies $\E[|p-q|]\le 4\cdot k$. Hence by Markov's inequality, 
\begin{equation}\label{eq:key_lemma_1}
\Pr\sbra{(u,v)\notin\lambda}\ge\Pr\sbra{p-q\le100\cdot k}=1-\Pr\sbra{p-q>100\cdot k}\ge1-\frac{4\cdot k}{100\cdot k}=0.96.
\end{equation}
On the other hand, by setting $C_1$ large enough we know from \Cref{thm:CGK} 
$$
\Pr\sbra{\#\text{progress steps in $\lambda$}\le C_1\cdot k^2}\ge 0.99.
$$
Hence, by setting $C_2$ large enough, we have
\begin{equation*}
\Pr\sbra{(u,v)\notin\lambda \bigwedge \#\text{progress steps in $\lambda$}\le C_1\cdot k^2}\ge0.96+0.99-1\ge\frac{1}{C_2\cdot k}.\tag*{\qedhere}
\end{equation*}
\end{proof}

\subsection[Proof of Lemma 3.6: Case |u-v|<=100k]{Proof of \Cref{lem:key_lemma}: Case $|u-v|\le 100\cdot k$}\label{sec:proof_of_lem:key_lemma_small_gap}

First we need the following definition.

\begin{definition}[Stable zone $\Zcal$, \cite{BelazzouguiZ16}]
The \emph{stable zone} $\Zcal$ of $(u,v)$ consists of substrings $x[u'..u],y[v'..v]$ of equal length $L=u-u'+1=v-v'+1$, where $L \le \min\{u,v\}$ is the maximum possible length satisfying $x[u'..u]=y[v'..v]$. In particular, $u-v=u'-v'$; and $(u',v')\neq(1,1)$ as $(u,v)\notin\Mcal$.

Moreover, we say a state $(p,q)$ \emph{enters} $\Zcal$ if $p\ge u'$ and $q\ge v'$.
\end{definition}

We will find the following claim useful. It is proved in \cite{BelazzouguiZ16}, and we give a slightly different proof here for completeness.

\begin{claim}[{\cite[Claim 21]{BelazzouguiZ16}}]\label{clm:non_vertical}
Consider an $\infty$-step CGK random walk $\lambda$ on $x,y$, where $p,q$ are the pointers on $x,y$ respectively. 
Let $T$ be the first time that $\lambda$ enters $\Zcal$, i.e., $(p_T\ge u')\land(q_T\ge v')$. 
Then 
$$
\Pr\sbra{p_T-q_T\neq u-v}=\Pr\sbra{p_T-q_T\neq u'-v'}\ge2/3.
$$
\end{claim}
\begin{proof}
Note that event ``$p_T-q_T=u'-v'$'' is exactly ``$(p_T,q_T)=(u',v')$''.
Hence one of the following three cases must happen at some step $t<T$ in order to make $(p_T,q_T)=(u',v')$ possible:
\begin{itemize}
\item $(p_t,q_t)=(u',v'-1)$. Then the transitions must satisfy $r(i,x[u'])=r(i,y[v'-1])=0$ for all $i\in[t..T-2]$ and $r(T-1,x[u'])=1,r(T-1,y[v'-1])=0$, which happens with probability $1/3$ if $x[u']\neq y[v'-1]$ and with probability $0$ if not.
\item $(p_t,q_t)=(u'-1,v')$. Similar analysis.
\item $(p_t,q_t)=(u'-1,v'-1)$. If the state after time $t$ does not fall into the previous two cases, the transitions must satisfy $r(i,x[u'-1])=r(i,y[v'-1])=0$ for all $i\in[t..T-2]$ and $r(T-1,x[u'-1])=r(T-1,y[v'-1])=1$, which happens with probability $1/3$. \qedhere
\end{itemize}
\end{proof}

We will also rely on the following technical result, the proof of which is in \Cref{sec:bounds_on_rho}.

\begin{proposition}\label{prop:bounds_on_rho}
There exists a universal constant $C_4\ge1$ such that the following holds. 
Assume $X,Y$ are two identical length-$L$ strings over alphabet $\Sigma$. 
Assume there exists a size-$M$ matching $(i_1,j_1),\ldots,(i_M,j_M)\in[L]^2$ such that
\begin{itemize}
\item $i_t>j_t$ and $X[i_t]=Y[j_t]$ hold for all $t\in[M]$;
\item $i_1<i_2<\cdots<i_M$ and $j_1<j_2<\cdots<j_M$.
\end{itemize}
Let $\rho=C_4\cdot(L-M)$ and $(\hat I,\hat J)$ be any state satisfying $\hat I-\hat J\ge\rho$.
Then a CGK random walk on $X,Y$ starting from $(\hat I,\hat J)$ will miss $(L,L)$ with probability at least $0.5$.
\end{proposition}

By symmetry, we derive the following corollary.

\begin{corollary}\label{cor:bounds_on_rho}
Let $C_4\ge1$ be the same constant in \Cref{prop:bounds_on_rho}. 
Assume $X,Y$ are two identical length-$L$ strings over alphabet $\Sigma$. 
Assume there exists a size-$d$ matching $(i_1,j_1),\ldots,(i_M,j_M)\in[L]^2$ such that
\begin{itemize}
\item $i_t>j_t$ holds for all $t\in[M]$, or $i_t<j_t$ holds for all $t\in[M]$;
\item $X[i_t]=Y[j_t]$ holds for all $t\in[M]$;
\item $i_1<i_2<\cdots<i_M$ and $j_1<j_2<\cdots<j_M$.
\end{itemize}
Let $\rho=C_4\cdot(L-M)$ and $(\hat I,\hat J)$ be any state satisfying $|\hat I-\hat J|\ge\rho$.
Then a CGK random walk on $X,Y$ starting from $(\hat I,\hat J)$ will miss $(L,L)$ with probability at least $0.5$.
\end{corollary}

\begin{proof}[Proof of \Cref{lem:key_lemma}: Case $|u-v|\le 100\cdot k$]
Let $C_5\ge1$ be a large constant. 
We will apply \Cref{prop:bounds_on_rho} with parameter $M\ge L-103\cdot k$; and let $\rho=C_4\cdot 103k$ be the corresponding bound in it.

We expect $\lambda$ to have the following three phases:
\begin{itemize}
\item $\Ecal_1$: $\lambda$ enters $\Zcal$ in a state $(p_1,q_1)$ within $C_5\cdot k^2$ progress steps, where $0<|(p_1-q_1)-(u-v)|\le200\cdot k$.
\item $\Ecal_2$: Starting from $(p_1,q_1)$ and within $2\cdot\rho^2$ progress steps, $\lambda$ reaches a state $(p_2,q_2)$ where either $(p_2,q_2)>(n,n)$ or $|(p_2-q_2)-(u-v)|\ge\rho$. Also, during the walk from $(p_1,q_1)$ to $(p_2,q_2)$, $\lambda$ never reaches some state $(p,q)$ satisfying $(p-q)-(u-v)=0$.
\item $\Ecal_3$: $(u,v)\notin\lambda$ and \#progress steps in $\lambda\le 2\cdot\rho^2+C_5\cdot\pbra{k^2+(\rho+301\cdot k)^2}$.
\end{itemize}

\begin{claim}\label{clm:key_lemma_E_1}
$\Pr\sbra{\Ecal_1}\ge0.5$.
\end{claim}
\begin{claim}\label{clm:key_lemma_E_2}
$\Pr\sbra{\Ecal_2\mid\Ecal_1}\ge1/(2\cdot\rho)$.
\end{claim}
\begin{claim}\label{clm:key_lemma_E_3}
$\Pr\sbra{\Ecal_3\mid\Ecal_1\land\Ecal_2}\ge1/4$.
\end{claim}

Assuming \Cref{clm:key_lemma_E_1}, \Cref{clm:key_lemma_E_2}, and \Cref{clm:key_lemma_E_3}, we show the following desired bound
$$
\Pr\sbra{(u,v)\notin\lambda\bigwedge\#\text{progress steps in $\lambda$}\le C_1\cdot k^2}
\ge\Pr\sbra{\Ecal_3}\ge\Pr\sbra{\Ecal_1\land\Ecal_2\land\Ecal_3}\ge\frac{1}{C_2\cdot k}
$$
by setting $C_1=2\cdot\pbra{103\cdot C_4}^2+C_5\cdot\pbra{1+(301+103\cdot C_4)^2}$ and $C_2=16$.
\end{proof}

\begin{proof}[Proof of \Cref{clm:key_lemma_E_1}]
By \Cref{thm:CGK} with $C_5$ large enough, we ensure 
$$
\Pr\sbra{\#\text{progress steps in $\lambda$}\le C_5\cdot k^2}\ge0.99.
$$
Then combining \Cref{clm:non_vertical} and \Cref{eq:key_lemma_1}, we have
\begin{align*}
&\phantom{\ge}\Pr\sbra{\Ecal_1}\\
&\ge\Pr\sbra{\#\text{progress steps in $\lambda$ before $\Zcal$}\le C_5\cdot k^2}\\
& \quad +\Pr\sbra{(p_1,q_1)\neq(u',v')}+\Pr\sbra{|(p_1-q_1)-(u-v)|\le200\cdot k}-2\\
&\ge\Pr\sbra{\#\text{progress steps in $\lambda$}\le C_5\cdot k^2}+\Pr\sbra{(p_1,q_1)\neq(u',v')}+\Pr\sbra{|p_1-q_1|\le100\cdot k}-2\\
&\ge0.99+2/3+0.96-2\\
&>0.5.\tag*{\qedhere}
\end{align*}
\end{proof}

\begin{proof}[Proof of \Cref{clm:key_lemma_E_2}]
Let $a=|(p_1-q_1)-(u-v)|$. Then by conditioning on $\Ecal_1$, we know $a\ge1$. If $a\ge\rho$, then $(p_2,q_2)=(p_1,q_1)$ and the claims holds immediately. Therefore we assume $a<\rho$. 

For convenience, if the case $(p_2,q_2)>(n,n)$ happens, we replace $x[p_2+1..\infty]$ and $y[q_2+1..\infty]$ with random strings\footnote{The random strings here are not essential. The only purpose of this is to generate infinitely many progress steps.} and continue the walk. 
Denote this new walk as $\lambda'$ and we stop it when $(p-q)-(u-v)=0$ or $|(p-q)-(u-v)|\ge\rho$ happens.
Let $\pi$ be a one-dimensional unbiased and self-looped random walk starting from $a$. Due to the correspondence between transitions in $\pi$ and progress steps in $\lambda'$, we can stop $\pi$ when it reaches $0$ or $\rho$; and hence
\begin{align*}
\Pr\sbra{\Ecal_2\mid\Ecal_1}
&\ge\Pr\sbra{\lambda'\text{ stops at the second case within }2\cdot\rho^2\text{ progress steps}}\\
&=\Pr\sbra{\pi\text{ stops at }\rho\text{ in }2\cdot\rho^2\text{ steps}}\\
&=\Pr\sbra{\pi\text{ stops at }\rho}\cdot\Pr\sbra{\pi\text{ stops }\rho\text{ in }2\cdot\rho^2\text{ steps}\middle|\pi\text{ hits }\rho\text{ first}}\\
&=\Pr\sbra{\pi\text{ stops at }\rho}\cdot\pbra{1-\Pr\sbra{\pi\text{ stops }\rho\text{ after }2\cdot\rho^2\text{ steps}\middle|\pi\text{ hits }\rho\text{ first}}}\\
&\ge\Pr\sbra{\pi\text{ stops at }\rho}\cdot\pbra{1-\frac{\E\sbra{\#\text{steps until }\pi\text{ stops}\middle|\pi\text{ hits }\rho\text{ first}}}{2\cdot\rho^2}}.
\tag{due to Markov's inequality}
\end{align*}
By \Cref{thm:steps_in_random_walk}, we have
\begin{align*}
a\cdot(\rho-a)
&=\E\sbra{\#\text{steps until }\pi\text{ stops}}\\
&\ge\Pr\sbra{\pi\text{ stops at }\rho}\E\sbra{\#\text{steps until }\pi\text{ stops}\middle|\pi\text{ stops at }\rho}\\
&=\frac a\rho\cdot\E\sbra{\#\text{steps until }\pi\text{ stops}\middle|\pi\text{ stops at }\rho}.
\end{align*}
Hence 
\begin{equation*}
\Pr\sbra{\Ecal_2\mid\Ecal_1}\ge\frac a\rho\cdot\pbra{1-\frac{\rho\cdot(\rho-a)}{2\cdot\rho^2}}\ge\frac1{2\cdot\rho}.\tag*{\qedhere}
\end{equation*}
\end{proof}

\begin{proof}[Proof of \Cref{clm:key_lemma_E_3}]
Observe that
$$
\Pr\sbra{\Ecal_3\mid\Ecal_1\land\Ecal_2}
\ge\Pr\sbra{\text{after }(p_2,q_2),~(u,v)\notin\lambda \bigwedge \#\text{progress steps}\le C_5\cdot(\rho+301\cdot k)^2}.
$$
If $(p_2,q_2)>(n,n)$ then there is nothing to prove. So we focus on the case $|(p_2-q_2)-(u-v)|\ge\rho$.

Let $X=Y=x[u'..u]=y[v'..v]$, $L=u-u'+1=v-v'+1$, and $(\hat I,\hat J)=(p_2-u'+1,q_2-v'+1)$. Then $|\hat I-\hat J|\ge\rho$.
Recall that $\Mcal$ is the greedy optimal matching between $x,y$. We retain those edges $(p,q)\in\Mcal$ that is inside $\Zcal$ to form a matching $\Mcal'$ between $X,Y$, i.e., 
$$
\Mcal'=\cbra{(p-u'+1,q-v'+1)\mid (p,q)\in\Mcal, u'\le p\le u, v'\le q\le v}.
$$
By \Cref{lem:optimal_matching}, we have
$$
M:=|\Mcal'|\ge L-3\cdot\ed(x,y)-|u-v|\ge L-103\cdot k.
$$
By \Cref{lem:greedy_optimal_matching} and $(u,v)\notin\Mcal$, the condition in \Cref{cor:bounds_on_rho} holds and we get 
$$
\Pr\sbra{\text{after }(p_2,q_2),~(u,v)\notin\lambda}\ge0.5.
$$

Therefore, it suffices to prove
\begin{equation}\label{eq:key_lemma_2}
\Pr\sbra{\text{after }(p_2,q_2),~\#\text{progress steps}\le C_5\cdot(\rho+301\cdot k)^2}\ge\frac34.
\end{equation}
Note that 
$$
\ed(x[p_2..\infty],y[q_2..\infty])\le\ed(x,y)+|p_2-q_2|\le\ed(x,y)+|(p_2-q_2)-(u-v)|+|u-v|.
$$
Thus \Cref{eq:key_lemma_2} follows from \Cref{thm:CGK} with $C_5$ large enough and the following estimate:
\begin{itemize}
\item If $|(p_1-q_1)-(u-v)|\le\rho$, then $|(p_2-q_2)-(u-v)|=\rho$. Hence 
$$
\ed(x[p_2..\infty],y[q_2..\infty])\le k+\rho+100\cdot k\le\rho+301\cdot k.
$$
\item If $\rho<|(p_1-q_1)-(u-v)|\le200\cdot k$, then $(p_2,q_2)=(p_1,q_1)$ and 
$|(p_2-q_2)-(u-v)|\le200\cdot k$. Hence
\begin{equation*}
\ed(x[p_2..\infty],y[q_2..\infty])\le k+200\cdot k+100\cdot k\le\rho+301\cdot k.\tag*{\qedhere}
\end{equation*}
\end{itemize}
\end{proof}

\subsection{Sketch Construction}
\label{sec:sketch-construction}
This section is devoted for the detailed description of \Cref{constr:sketch_for_each_random_walk} and its correctness. The construction is mostly based on \cite[Section 4.1]{BelazzouguiZ16}, with a few simplifying modifications. 

\begin{proof}[Proof of \Cref{constr:sketch_for_each_random_walk}]
Let $x',y'\in\Sigma^m$  be the outputs of the random walks on $x$ and $y$ respectively, i.e., $x'=\lambda_r(x)$ and $y'=\lambda_r(y)$, where $r$ is the (public) randomness. 

\paragraph*{\fbox{The encoding algorithm.}}
We only describe the encoding algorithm for $x$. The algorithm for $y$ is analogous.

Build a full binary tree $\Tcal_x$ of depth $d=\log m$ on top of $x'$, where the segments in depth $i$ have length $m/2^i$.
We use $\Tcal_x(i,j)$ to denote the $j$-th segment in depth $i$. More precisely, for $i\in[0..d],j\in[2^i]$,  
$$
\Tcal_x(i,j)=x'\sbra{1+\frac{(j-1)\cdot m}{2^i}..\frac{j\cdot m}{2^i}}.
$$
Let $f:\Sigma^{\le m}\to\Gamma$ be a randomized hash function\footnote{The same $f$ is used for both $x$ and $y$ in this random walk $\lambda$, but it may be different among different walks.} in \Cref{thm:rolling_hash} by setting $\eta_1=\eta/(2(d+1)\cdot m)$, and define $h^x_{i,j}=f(\Tcal_x(i,j))$.

\begin{theorem}[Rolling hash, \cite{KarpR87}]\label{thm:rolling_hash}
Let $\eta_1\in(0,1)$.
There exists a randomized hash function $f\colon\Sigma^{\le m}\to\Gamma$ where
\begin{enumerate}[label=(\alph*)]
\item $|\Gamma|=O(|\Sigma|+(m/\eta_1))$ and $\Sigma\subseteq\Gamma$;
\item for any distinct $z,z'\in\Sigma^{\le m}$, $\Pr\sbra{f(z)=f(z')}\le\eta_1$; and for any single character $c\in\Sigma$, $f(c)=c$;\footnote{When we interpret $\Sigma$ as a set of numbers, the rolling hash is $f(x)=\sum_{i=1}^{|x|}x_i\cdot r^{i-1}\mod p$, where $r$ is a random number and $p$ is a large prime. Therefore $f(\cdot)$ is an identity map on a single character.}
\item $f$ can be computed with $O(\log|\Gamma|)$ bits of space and $\poly(\log|\Gamma|)$ time per character.
\end{enumerate}
\end{theorem}

In addition, we define the following quantities for each segment $\Tcal_x(i,j)$ which will be used later to identify an effective alignment:
\begin{itemize}
\item let $\ell^x_{i,j}\in[m]$ be the length of the pre-image of $\Tcal_x(i,j)$ in $x[1..\infty]$; 
\item let $\alpha^x_{i,j}\in\bin$ be the indicator of whether the pre-image of the first character of $\Tcal_x(i,j)$ and the last character of $\Tcal_x(i,j-1)$\footnote{If $j=1$, $\Tcal_x(i,j-1)$ is not well-defined. Then we simply define $\alpha^x_{i,j}=0$ in this case.} are identical;
\item let $\beta^x_{i,j}\in\bin$ be the indicator of whether the pre-image of the last character of $\Tcal_x(i,j)$ and the first character of $\Tcal_x(i,j+1)$\footnote{If $j=2^i$, $\Tcal_x(i,j+1)$ is not well-defined. Then we simply define $\beta^x_{i,j}=0$ in this case.} are identical.
\end{itemize}
That is, assume the pointers are $p_1,p_2,p_3,p_4$ for $\Tcal_x(i,j-1)$'s last character, $\Tcal_x(i,j)$'s first character, $\Tcal_x(i,j)$'s last character, and $\Tcal_x(i,j+1)$'s first character respectively. Then
\begin{itemize}
\item $\ell^x_{i,j}=p_3-p_2+1$.
\item If $p_1=p_2$ then $\alpha^x_{i,j}=1$; otherwise $\alpha^x_{i,j}=0$.
\item If $p_3=p_4$ then $\beta^x_{i,j}=1$; otherwise $\beta^x_{i,j}=0$.
\end{itemize}

Let $U_x$ be the set of $6$-tuples $(i,j,h^x_{i,j},\ell^x_{i,j},\alpha^x_{i,j},\beta^x_{i,j})\in[0..d]\times[m]\times\Gamma\times[m]\times\bin\times\bin$.
We write $u_x$ as the indicator vector for $U_x$, which is a binary vector of length $4(d+1)\cdot m^2\cdot|\Gamma|$. 
We now apply \Cref{thm:sparse_recovery_sketch}\footnote{We use a more recent result \cite{KapralovNPWWY17} instead of \cite{PoratL07} as in \cite{BelazzouguiZ16}, since it is easier to state our dependence on the failure probability as an independent parameter. The result of \cite{PoratL07} is stated only for failure probability $1/n$. It may be possible to alter their result to also have failure probability stated as an independent parameter, but we have not verified this, and citing \cite{KapralovNPWWY17} instead allowed us to bypass doing so.} with $L=4(d+1)\cdot m^2\cdot|\Gamma|,\Delta=6(d+1)\cdot Ck^2,\eta_2=\eta/2$ to obtain the sketch $\Scal(u_x)$.

\begin{theorem}[{\cite[Section A.3]{KapralovNPWWY17}}]\label{thm:sparse_recovery_sketch}
Let $L,\Delta$ be two positive integers and $\eta_2\in(0,1)$. Let $z,z'\in\bin^L$. There exists an efficient randomized linear sketching algorithm $\Scal$ where
\begin{itemize}
\item the sketch size and the encoding space are $O(\Delta\log L+\log(L/\eta_2))$;
\item the encoding time per bit and the decoding time are $\poly(\Delta\log(L/\eta_2))$;
\item if $z-z'$ has at most $\Delta$ non-zero coordinates, the decoder recovers $z-z'$ exactly;
\item if $z-z'$ has more than $\Delta$ non-zero coordinates, the decoder reports ``fail'' with probability at least $1-\eta_2$.
\end{itemize}
\end{theorem}

The final sketch is $(\ell^x_{0,1},\Scal(u_x))$. Since $\ell^x_{0,1}\in[m]$ is the length of the pre-image of the whole $\lambda$, it in particular indicates if $\lambda$ walks through $x$.

\paragraph*{\fbox{The decoding algorithm.}}
Let $\Scal_x=(\ell^x_{0,1},\Scal(u_x))$ and $\Scal_y=(\ell^y_{0,1},\Scal(u_y))$ be the sketches for $x,y$ in this round of random walk respectively. We first condition on the event $\Ecal$ that the following two bullets simultaneously occur:
\begin{itemize}
\item For all $i,j$, if $\Tcal_x(i,j)\neq\Tcal_y(i,j)$ then $h^x_{i,j}\neq h^y_{i,j}$.
\item If $u_x$ and $u_y$ differ in more than $\Delta$ coordinates, the decoder from \Cref{thm:sparse_recovery_sketch} reports ``fail''.
\end{itemize}
By a union bound, this conditioning costs at most $\eta$ probability loss. 

If $\ell^x_{0,1}<n$ or $\ell^y_{0,1}<n$ or the decoder from \Cref{thm:sparse_recovery_sketch} reports ``fail'', we directly report ``error''. Let $\Ecal'$ be the event that we haven't reported ``error''. Then it suffices to verify the following two claims.

\begin{claim}
Conditioning on $\Ecal$ and $\lambda$ containing at most $C\cdot k^2$ progress steps, $u_x$ and $u_y$ differ in at most $\Delta$ coordinates.
\end{claim}
\begin{proof}
Let 
\begin{align*}
&V_1=\cbra{(i,j)\mid \Tcal_x(i,j)\neq\Tcal_y(i,j)},
&V_2=\cbra{(i,j)\mid \ell^x_{i,j}\neq\ell^y_{i,j}},\\
&V_3=\cbra{(i,j)\mid \alpha^x_{i,j}\neq\alpha^y_{i,j}},
&V_4=\cbra{(i,j)\mid \beta^x_{i,j}\neq\beta^y_{i,j}}.
\end{align*}
Conditioning on $\Ecal$, the number of coordinates where $u_x,u_y$ differ equals twice the size of $V_1\cup V_2\cup V_3\cup V_4$.
\begin{itemize}
\item \textbf{For $V_1$.} Observe that $(i,j)\in V_1$ iff there is some progress step among $1+(j-1)\cdot m/2^i$ and $j\cdot m/2^i$. Hence by the structure of the full binary tree, we have $|V_1|\le(d+1)\cdot\#\text{progress steps in $\lambda$}\le(d+1)\cdot Ck^2$.
\item \textbf{For $V_2$.} By Item (2) in \Cref{thm:CGK}, if $\Tcal_x(i,j)=\Tcal_y(i,j)$ then their pre-image equals, which implies $\ell^x_{i,j}=\ell^y_{i,j}$. Thus $V_2\subseteq V_1$.
\item \textbf{For $V_3$.} By Item (2) in \Cref{thm:CGK}, if $\Tcal_x(i,j-1)\circ\Tcal_x(i,j)=\Tcal_y(i,j-1)\circ\Tcal_y(i,j)$ then their pre-image equals, which implies the random walk on this part is identical for $x,y$ and hence $\alpha^x_{i,j}=\alpha^y_{i,j}$. Thus $V_3\subseteq V_1\cup\cbra{(i,j)\mid (i,j-1)\in V_1}$.
\item \textbf{For $V_4$.} Similar analysis as for $V_3$. We have $V_4\subseteq V_1\cup\cbra{(i,j)\mid (i,j+1)\in V_1}$.
\end{itemize}
In all, $|V_1\cup V_2\cup V_3\cup V_4|=|V_1\cup V_3\cup V_4|\le|V_1|+|V_3\setminus V_1|+|V_4\setminus V_1|\le 3\cdot |V_1|\le 3(d+1)\cdot Ck^2=\Delta/2$.
\end{proof}

\begin{claim}
Conditioning on $\Ecal,\Ecal'$, we can compute in $\poly(k\log(n/\eta))$ time an effective alignment $\Acal$ consistent with $\lambda$.
\end{claim}
\begin{proof}
Conditioning on $\Ecal,\Ecal'$, we use the decoder of \Cref{thm:sparse_recovery_sketch} to mark those $(i,j)$ that 
$$
(h^x_{i,j},\ell^x_{i,j},\alpha^x_{i,j},\beta^x_{i,j})\neq(h^y_{i,j},\ell^y_{i,j},\alpha^y_{i,j},\beta^y_{i,j}).
$$
We construct $\Acal$ by performing a DFS on the tree as in \Cref{alg:construct_effective_alignment}. Its correctness is guaranteed by $\ell^x_{0,1},\ell^y_{0,1}\ge n$.

\begin{algorithm}[ht]
\caption{Construct effective alignment $\Acal$}\label{alg:construct_effective_alignment}
\DontPrintSemicolon
\LinesNumbered
\SetKwProg{proc}{Procedure}{}{}
\SetKwFunction{DFS}{DFS}
\KwIn{$\ell_{0,1}^x,\ell_{0,1}^y$ and $(h^x_{i,j},\ell^x_{i,j},\alpha^x_{i,j},\beta^x_{i,j}),(h^y_{i,j},\ell^y_{i,j},\alpha^y_{i,j},\beta^y_{i,j})$ for marked nodes $(i,j)$}
\KwOut{An effective alignment $\Acal$ consistent with $\lambda$}
Initialize $\Acal=(G,g_x,g_y)$ with an empty graph $G$ and empty functions $g_x,g_y$\;
\DFS{$1,0,1,\ell^x_{0,1},1,\ell^y_{0,1}$}\;
\proc{\DFS{$i,j,s_x,e_x,s_y,e_y$}}{
\tcc{$(i,j)$ is the current node on the binary tree.
$s_x$ (resp., $s_y$) is the pointer for $\Tcal_x(i,j)$'s (resp., $\Tcal_y(i,j)$'s) first character.
$e_x$ (resp., $e_y$) is the pointer for $\Tcal_x(i,j)$'s (resp., $\Tcal_y(i,j)$'s) last character.}
\eIf{$(i,j)$ is not marked}{
\ForEach{edge $(p_x,p_y)\in\cbra{(s_x,s_y),(s_x+1,s_y+1),\ldots,(e_x,e_y)}$}{
\lIf{$p_x\le n$ and $p_y\le n$}{Add edge $(p_x,p_y)$ to $G$}
\lIf{$p_x\le n$ and $p_y>n$}{Update $g_x(p_x)=0$}
\lIf{$p_x>n$ and $p_y\le n$}{Update $g_y(p_y)=0$}
}
}{
\eIf{$(i,j)$ is a leaf}{
Update $g_x(s_x)=h^x_{i,j}$ and $g_y(s_y)=h^y_{i,j}$
\tcp*[f]{due to Item (b) in \Cref{thm:rolling_hash}}
}{
$(i_1,j_1)\gets$ left child of $(i,j)$ and $(i_2,j_2)\gets$ right child of $(i,j)$\;
\eIf{$(i_1,j_1)$ is marked}{
$m_x\gets s_x+\ell^x_{i_1,j_1}-1$ and $m_y\gets s_y+\ell^y_{i_1,j_1}-1$\;
$o_x\gets 1-\beta^x_{i_1,j_1}$ and $o_y\gets 1-\beta^y_{i_1,j_1}$\;
\DFS{$i_1,j_1,s_x,m_x,s_y,m_x$}\;
\DFS{$i_2,j_2,m_x+o_x,e_x,m_y+o_y,e_y$}
}(\tcp*[f]{$(i_2,j_2)$ must be marked}){
$m_x\gets e_x-\ell^x_{i_2,j_2}+1$ and $m_y\gets e_y-\ell^y_{i_2,j_2}+1$\;
$o_x\gets 1-\alpha^x_{i_1,j_1}$ and $o_y\gets 1-\alpha^y_{i_1,j_1}$\;
\DFS{$i_1,j_1,s_x,m_x-o_x,s_y,m_x-o_x$}\;
\DFS{$i_2,j_2,m_x,e_x,m_y,e_y$}
}
}
}
}
\end{algorithm}
Observe that the number of recursions is at most the number of marked nodes, which is at most $\Delta$. Hence the running time of \Cref{alg:construct_effective_alignment} is $\poly(k\log(n/\eta))$.
\end{proof}

\paragraph*{\fbox{Bounds on the parameters.} }
The decoding time is obvious, so we only calculate the parameters for encoding.
By \Cref{thm:sparse_recovery_sketch}, the sketch size is 
$$
O(\log m)+O(\Delta\log L+\log(L/\eta_2))=O\pbra{k^2\log(n/\eta)\log n}.
$$
When doing encoding, we work on each depth of $\Tcal$ in parallel. For a fixed depth, the hashes can be computed sequentially. Hence by \Cref{thm:rolling_hash} and \Cref{thm:sparse_recovery_sketch}, the encoding space\footnote{We omit the space for storing auxiliary information (e.g., current nodes) in the calculation, since these are minor term.} is 
$$
(d+1)\cdot O(\log|\Gamma|)+O(\Delta\log L+\log(L/\eta_2))=O\pbra{k^2\log(n/\eta)\log n}.
$$
Note that when a character arrives, we generate at most $(d+1)$ $6$-tuples, hence the encoding time per character is bounded by
\begin{equation*}
(d+1)\cdot\pbra{\poly(\log|\Gamma|)+\poly(\Delta\log(L/\eta_2))}=\poly(k\log(n/\eta)).\tag*{\qedhere}
\end{equation*}
\end{proof}
\section{CGK Random Walks on Self-similar Strings}\label{sec:bounds_on_rho}

This section is devoted to the proof of \Cref{prop:bounds_on_rho}. It characterizes CGK random walks on strings of certain self-similarity, which may be interesting on its own.

\begin{proposition*}[\Cref{prop:bounds_on_rho} restated]
There exists a universal constant $C_4\ge1$ such that the following holds. 
Assume $X,Y$ are two identical length-$L$ strings over alphabet $\Sigma$. 
Assume there exists a size-$M$ matching $(i_1,j_1),\ldots,(i_M,j_M)\in[L]^2$ such that
\begin{itemize}
\item $i_t>j_t$ and $X[i_t]=Y[j_t]$ hold for all $t\in[M]$;
\item $i_1<i_2<\cdots<i_M$ and $j_1<j_2<\cdots<j_M$.
\end{itemize}
Let $\rho=C_4\cdot(L-M)$ and $(\hat I,\hat J)$ be any state satisfying $\hat I-\hat J\ge\rho$.
Then a CGK random walk on $X,Y$ starting from $(\hat I,\hat J)$ will miss $(L,L)$ with probability at least $0.5$.
\end{proposition*}

We will first provide necessary definitions and explore basic properties in \Cref{sec:stable-defn}. Then we relate them with edit distances in \Cref{sec:stable-edit}, and present the main proof in \Cref{sec:stable-and} and \Cref{sec:together}. The proof of a technical lemma is deferred to \Cref{sec:stabilize}.

\subsection{Stable States}
\label{sec:stable-defn}

We fix the matching in \Cref{prop:bounds_on_rho}, so when we say $(i,j)$ is a matched edge it means $(i,j)$ is an edge in the matching. 
We extend $X,Y$ to $X[-\infty..\infty],Y[-\infty..\infty]$ by adding dummy characters $X[i]=Y[i]=X[L]$ for all $i>L$, and $X[i]=Y[i]=X[1]$ for all $i<1$. We also add matched edges $(i,i-1)$ for all $i>L$ as well as $i\le 1$. Note that all the edges are still non-intersecting. 
Though the added characters may not be consistent with the original input strings $x,y$, it does not change the probability of the walk missing $(L,L)$.
Since $X=Y$ and the initial state satisfies $\hat I\ge \hat J+\rho\ge \hat J$, any future state $(I,J)$ must still satisfy $I\ge J$.

We introduce the notion of \emph{stable segment}. 
\begin{definition}[Stable segment]
We say $[l..r]$ is a \emph{stable segment}, if for every matched edge $(I,J)$ (where we must have $I>J$), exactly one of the following two conditions hold:
\begin{itemize}
    \item $J<l$ and $I\le r$.
    \item $J\ge l$ and $I>r$.
\end{itemize}
\end{definition}

\begin{figure}[ht]
    \centering
    \begin{tikzpicture}
    \node[scale=1] at (0,60 pt) {index};
    \foreach \i in {0,1,2,3,4,5,6,7,8,9,10}
        \node[scale=1] at (\i*20 + 60 pt, 60 pt) {$\i$};

    \foreach \i/\p in {0/$a$,1/$a$,2/$c$,3/$a$,4/$b$,5/$c$,6/$a$,7/$b$,8/$a$,9/$b$,10/$b$}
    {
        \node at (\i*20 + 60 pt, 40 pt) {$\phantom{b}$\p$\phantom{b}$};
    }

    \node at (0, 20pt) {$X$};

    \node at (30 pt, 0pt) {$\ldots\ldots$};
    \node at (280 pt, 0 pt) {$\ldots\ldots$};

    \node at (0, -20 pt) {$Y$};

    \draw[blue,dashed,very thick] (70pt, -30pt) -- (70pt, 50pt);
    \draw[blue,dashed,very thick] (110pt, -30pt) -- (110pt, 50pt);
    \draw[blue,dashed,very thick] (170pt, -30pt) -- (170pt, 50pt);
    \draw[blue,dashed,very thick] (210pt, -30pt) -- (210pt, 50pt);
    \draw[blue,dashed,very thick] (250pt, -30pt) -- (250pt, 50pt);

    \draw[red,very thick] (60pt, -20pt) -- (80pt, 20pt);   
    \draw[red,very thick] (80pt, -20pt) -- (120pt, 20pt);   
    \draw[red,very thick] (100pt, -20pt) -- (160pt, 20pt);   
    \draw[red,very thick] (120pt, -20pt) -- (180pt, 20pt);   
    \draw[red,very thick] (140pt, -20pt) -- (200pt, 20pt);   
    \draw[red,very thick] (180pt, -20pt) -- (220pt, 20pt);
    \draw[red,very thick] (200pt, -20pt) -- (240pt, 20pt);       
    \draw[red,very thick] (240pt, -20pt) -- (260pt, 20pt);   
    
    \foreach \i in {0,1,2,3,4,5,6,7,8,9,10}
        \draw[fill=black] (\i*20 + 60 pt, 20 pt) circle (2pt);
    \foreach \i in {0,1,2,3,4,5,6,7,8,9,10}
        \draw[fill=black] (\i*20 + 60 pt, -20 pt) circle (2pt);
\end{tikzpicture}
    \caption{A stable partition for $X[1..L]=Y[1..L]={acabcabab}$ ($L=9$).}
    \label{fig:stable-state}
\end{figure}
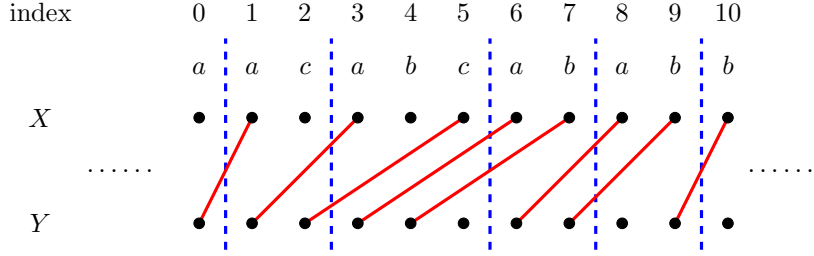

For example in \Cref{fig:stable-state}, every segment separated by blue dashed lines is a stable segment. 

\begin{remark}
To gain a better intuition of the definition, consider the special case where the string $X[1..L]$ has period $p$ and every matched edge $(I,J)$ inside segment $[1..L]$ satisfies $I-J=p$. In this periodic case, a segment contained in $[2..L-1]$ is stable if and only if its length is $p$. 

Our motivation is that, when there are few unmatched characters, using our more generalized definition we can approximately preserve the nice properties of periodic strings. 
For example, when $X$ has period $p$, the strings $X[i..i+tp-1]$ and $X[i+tp..i+2tp-1]$ must be identical. In a non-periodic case, we can similarly prove that $X[i..j-1]$ and $X[j..j+(j-i)-1]$ have small edit distance if $[i..j-1]$ can be divided into several stable segments (see \Cref{lem:equallength} for a more formal statement). In the remaining part of the section, readers are encouraged to use the periodic case for a more intuitive understanding.
\end{remark}

The following lemma says that every character can be the beginning/ending of some stable segment.
\begin{lemma} \label{lem:existseg}
The following hold.
\begin{enumerate}[label=(\arabic*)]
    \item For every $r\in \Zbb$, there exists some $l\le r$ such that $[l..r]$ is a stable segment.
    \item For every $l\in \Zbb$, there exists some $r \ge l$ such that $[l..r]$ is a stable segment.
\end{enumerate}
\end{lemma}
\begin{proof}
We prove Item (1), and the proof of Item (2) is analogous.  Let 
\[
l' = \min \cbra{l': \text{there exists $r'\ge r+1$ such that $(r',l')$ is a matched edge}}.
\]
We claim $\sbra{\min\cbra{r,l'}..r}$ must be a stable segment. If not, then there are two possible cases:
\begin{itemize}
    \item There exists a matched edge $(I,J)$ such that $J<l'$ and $I>r$. This contradicts the minimality of $l'$.
    \item  There exists a matched edge $(I,J)$ such that $l' \le J<I\le r$. By the definition of $l'$, there is another matched edge $(r',l')$ where $r'\ge r+1>I$. Since these two edges are non-intersecting, we must have $l'>J$. A contradiction.\qedhere
\end{itemize}
\end{proof}

\begin{definition}[Stable partition $\Pcal$ and stable states]
Consider a partition $\Pcal=\pbra{\Pcal_i}_i$ of the integers into segments, where $\Pcal_i=[p_i..p_{i+1}-1]$ and $p_i<p_{i+1}$. We say $\Pcal$ is a \emph{stable partition} if every $\Pcal_i$ is a stable segment. Then we say
\begin{itemize}
\item state $(I,J)$ is a \emph{$(\Pcal,b)$-stable state}, if there exists some $i$ such that $J=p_i$ and $I=p_{i+b}$;
\item state $(I,J)$ is a \emph{$b$-stable state}, if there exists a stable partition $\Pcal$ such that $(I,J)$ is a $(\Pcal,b)$-stable state;
\item state $(I,J)$ is a \emph{stable state}, if there exists some $b\ge0$ such that $(I,J)$ is a $b$-stable state. In particular, when $I\ge J>L$, $(I,J)$ is always a stable state.
\end{itemize}
\end{definition}

Note that, given a partition of $[J..I-1]$ into several stable segments, we can apply \Cref{lem:existseg} and extend it into a stable partition. Hence we have the following fact.

\begin{fact}
For any $b\ge1$, a state $(I,J)$ is a $b$-stable state iff $[J..I-1]$ can be partitioned into $b$ stable segments.
Moreover, a state $(I,J)$ is a 0-stable state iff $I=J$.
\end{fact}

From the definition of stable partition, we immediately have the following result.
\begin{proposition}
\label{prop:adjacent}
Let $\pbra{\Pcal_i}_i$ be a stable partition. Then for every matched edge $(I,J)$, there exists some $i$ such that $J\in \Pcal_{i-1}$ and $I\in \Pcal_{i}$.
\end{proposition}
\begin{proof}
Let $i,j$ be such that $I\in \Pcal_i$ and $J\in \Pcal_j$. Since $I>J$, we must have $i\ge j$. If $i>j+1$, then $\Pcal_{j+1}$ cannot be a stable segment. If $i=j$, then $\Pcal_j$ cannot be a stable segment. Hence we must have $i=j+1$.
\end{proof}

\begin{lemma}[Stable predecessors $\pred_L(\cdot),\pred_R(\cdot)$]
\label{lem:pred}
For every $I\in \Zbb$, there exists $\pred_{L}(I)\le \pred_{R}(I) \le I-1$ such that $[J..I-1]$ is a stable segment iff $\pred_{L}(I)\le J\le \pred_{R}(I)$.

Moreover, the following hold:
\begin{enumerate}[label=(\arabic*)]
\item $\pred_{L}(I)\le \pred_{L}(I+1)$;
\item $\pred_{R}(I)\le \pred_{R}(I+1)$.
\end{enumerate}
\end{lemma}
\begin{proof}
To prove the first part, suppose $J_0\le  J\le J_1$, and both $[J_0..I-1]$ and $[J_1..I-1]$ are stable segments. If $[J..I-1]$ is not a stable segment, then there are two possibilities:

\begin{itemize}
\item There is a matched edge $(I',J')$ such that $J'<J$ and $I-1<I'$. Then $J'<J_1\le I-1<I'$, contradicting that $[J_1..I-1]$ is a stable segment.
\item There is a matched edge $(I',J')$ such that $J\le J'<I'\le I-1$. Then $J_0< J'< I'\le I_1$, contradicting that $[J_0..I-1]$ is a stable segment.
\end{itemize}
Hence, $[J..I-1]$ is also a stable segment, which means $\pred_L(I),\pred_R(I)$ are well-defined.

We prove Item (1). Let $J_0 = \pred_L(I)$, then there must be a matched edge $(I',J_0-1)$ where $J_0\le I'\le I-1$.
If $\pred_L(I+1)<J_0$, then $[\pred_L(I+1)..I]$ is a stable segment, which contains matched edge $(I',J_0-1)$. A contradiction.

Now we turn to Item (2). Let $J_1 = \pred_R(I+1)$. Suppose that $\pred_R(I)>J_1$. Then $[J_1+1..I-1]$ is also a stable segment due to Item (1). By the definition of $J_1$, there must be a matched edge $(I',J_1)$ where $I'\ge I+1$. This contradicts that $[J_1+1..I-1]$ is a stable segment.
\end{proof}

For example in \Cref{fig:stable-state}, we have $\pred_{L}(8)=5, \pred_{R}(8)=6$.

\begin{lemma}[$b$-stable predecessors $\pred^{(b)}_L(\cdot),\pred^{(b)}_R(\cdot)$]
\label{lem:b-pred}
Define 
$$
\pred_{L}^{(b)}(I) = \begin{cases}
I&b=0,\\
\pred_{L}(\pred_{L}^{(b-1)}(I))&b\ge 1,
\end{cases}
\quad\text{and}\quad
\pred_{R}^{(b)}(I) = \begin{cases}
I&b=0,\\
\pred_{R}(\pred_{R}^{(b-1)}(I))&b\ge 1.
\end{cases}
$$ 
Then the following hold for all $I,b$:
\begin{enumerate}[label=(\arabic*)]
\item  $\pred_{L}^{(b)}(I)\le \pred_{L}^{(b)}(I+1)$, and $\pred_{R}^{(b)}(I)\le \pred_{R}^{(b)}(I+1)$, and $\pred^{(b)}_L(I)\le \pred^{(b)}_R(I)$.
\item  For all $J\le I-1$, $(I,J)$ is a $b$-stable state iff $\pred_L^{(b)}(I) \le J \le \pred_R^{(b)}(I)$.
\item $\pred_{R}^{(b)}(I)\ge \pred_{L}^{(b)}(I+1)-1$.
\end{enumerate}
\end{lemma}
\begin{proof}
Item (1) follows immediately from \Cref{lem:pred}.

We now prove Item (2) and (3) together by induction on $b$. Suppose both statements hold for $b-1$, where $b\ge 1$. 

We first prove the ``only if'' part of Item (2). We partition $[J..I-1]$ into $b$ stable segments, the last of which is $[I_1..I-1]$. Then we have $\pred_{L}^{(b-1)}(I_1) \le J \le\pred_{R}^{(b-1)}(I_1)$ by induction hypothesis on Item (2), and $\pred_{L}(I)\le I_1 \le \pred_R(I)$. Then we have $J\ge \pred_{L}^{(b-1)}(I_1) \ge \pred_{L}^{(b-1)}(\pred_{L}(I)) = \pred_L^{(b)}(I)$, and similarly $J\le \pred_R^{(b)}(I)$.

Now we prove the ``if'' part of Item (2). We show there exists some $I_1$ such that $\pred_L(I)\le I_1 \le  \pred_R(I)$ and $\pred_{L}^{(b-1)}(I_1) \le J \le \pred_{R}^{(b-1)}(I_1)$, which by induction hypothesis on Item (2) implies a partition of the segment $[J..I-1]$ into $b$ stable segments. 
Suppose there is no such $I_1$, then by the monotonicity of $\pred_{L}^{(b-1)}$ and $\pred_{R}^{(b-1)}$, there must exist $\pred_{L}(I)\le I'<\pred_{R}(I)$ such that $\pred_R^{(b-1)}(I') < J < \pred_L^{(b-1)}(I'+1)$, contradicting the induction hypothesis on Item (3).

Finally we prove Item (3). Suppose there exists some $I$ such that $\pred_R^{(b)}(I)< \pred_L^{(b)}(I+1)-1$. By \Cref{lem:existseg}, there exists $I'$ such that $[\pred_R^{(b)}(I)+1..I'-1]$ can be partitioned into $b$ stable segments. Hence by Item (2), $\pred_{L}^{(b)}(I') \le \pred_R^{(b)}(I)+1 \le \pred_{R}^{(b)}(I')$. Then, $\pred_{L}^{(b)}(I') \le \pred_{R}^{(b)}(I)+1 < \pred_L^{(b)}(I+1)$ implies $I'<I+1$; while $\pred_R^{(b)}(I') \ge \pred_{R}^{(b)}(I)+1$ implies $I'>I$. A contradiction.
\end{proof}

Given a stable partition $\Pcal$, we can define a predecessor function for $\Pcal$ as follows.
\begin{lemma}[Stable predecessor for a stable partition]
\label{lem:predpart}
Let $\Pcal=(\Pcal_i)_i$ be a stable partition where $\Pcal_i=[p_i..p_{i+1}-1]$ and $p_i<p_{i+1}$.
Then there exists a non-decreasing function $\pred_{\Pcal}\colon \Zbb \to \Zbb$ such that the following hold:
\begin{itemize}
\item For every $i$, $\pred_{\Pcal}(p_{i+1})=p_i$.
\item For every $I$, we have $\pred_{\Pcal}(I) \le I-1$, and  $[\pred_{\Pcal}(I).. I-1]$ is a stable segment. 
\end{itemize}
\end{lemma}
\begin{proof}
For any fixed $i$, it suffices to determine the values of $\pred_{\Pcal}(p_i+j)$ for $j=1,2,\dots,p_{i+1}-p_{i}-1$.

We simply define $\pred_{\Pcal}(p_i+j) = \max \cbra{\pred_{\Pcal}(p_i+j-1), \pred_L(p_i+j)}$. By \Cref{lem:pred} and induction on $j$, we have $\pred_L(p_i+j)\le \pred_{\Pcal}(p_i+j)\le \pred_R(p_i+j)$.
\end{proof}

Similarly as in \Cref{lem:b-pred}, we can define the $b$-stable predecessor $\pred_{\Pcal}^{(b)}(\cdot)$ for a stable partition $\Pcal$, which is sandwiched between $\pred_L^{(b)}(\cdot)$ and $\pred_R^{(b)}(\cdot)$, and also satisfies the Item (1) in \Cref{lem:b-pred}. As an example, in \Cref{fig:stable-state} $\pred_{\Pcal}^{(3)}(8)=1$.

\begin{corollary}[$b$-stable predecessors $\pred^{(b)}_\Pcal(\cdot)$] 
\label{cor:b-pred-P}
Let $\Pcal$ be a stable partition.
Define 
$$
\pred_\Pcal^{(b)}(I) = \begin{cases}
I&b=0,\\
\pred_\Pcal(\pred_\Pcal^{(b-1)}(I))&b\ge 1.
\end{cases}
$$ 
Then for any $I,b$, we have $\pred_\Pcal^{(b)}(I)\le \pred_\Pcal^{(b)}(I+1)$.

Moreover, $(I,\pred_\Pcal^{(b)}(I))$ is a $(\Pcal,b)$-stable state and $\pred_L^{(b)}(I)\le\pred_\Pcal^{(b)}(I)\le\pred_R^{(b)}(I)$.
\end{corollary}

\subsection{Edit Distances for Stable States}\label{sec:stable-edit}

We will bound the edit distance between stable states using the number of singletons. 

\begin{definition}[Singleton]
Every unmatched $X[i]$ or $Y[j]$ is called a \emph{singleton}.

Let $\sing_X[l,r)$ (resp., $\sing_Y[l,r)$) denote the number of singletons in $X[l..r-1]$ (resp., $Y[l..r-1]$). Let $\sing[l,r):=\sing_X[l,r)+\sing_Y[l,r)$. 
\end{definition}

We start with 1-stable states, i.e., stable segments.

\begin{lemma}\label{lem:compare}
Let $\Pcal$ be a stable partition.
For $I< I'$, let $J= \pred_{\Pcal}(I), J'= \pred_{\Pcal}(I')$.
Then
\begin{enumerate}[label=(\alph*)]
\item $\ed(X[I..I'-1], Y[J..J'-1]) \le \sing_{X}[I,I') + \sing_{Y}[J,J') \le \sing[J,I')$;
\item $|(I'-J') - (I-J)| = |(I'-I) - (J'-J)| \le \sing_{X}[I,I') + \sing_{Y}[J,J')\le \sing[J,I')$.
\end{enumerate}
\end{lemma}
\begin{proof}
For every fixed $i\in [I,I']$, let $j_0= \pred_{\Pcal}(i),j_1= \pred_{\Pcal}(i+1)$.
By \Cref{lem:predpart}, we observe that $X[i]$ is either unmatched or matched to one of the characters $Y[j_0],Y[j_0+1],\dots,Y[j_1-1]$; and also conversely these characters can only be matched to $X[i]$. Hence,
\[ \ed(X[i..i], Y[j_0..j_1-1]) \le \sing_{X}[i,i+1) + \sing_{Y}[j_0,j_1).  \]
Summing up over all $i= I,I+1,\dots,I'$, we obtain
$$
\ed(X[I..I'-1], Y[J..J'-1]) \le \sing_{X}[I,I') + \sing_{Y}[J,J'). 
$$
On the other hand, by \Cref{fct:length_to_edit} we have $\ed(X[I..I'-1],Y[J..J'-1])\ge\abs{(I'-I)-(J'-J)}$. Hence Item (b) follows immediately from Item (a).
\end{proof}

\begin{corollary}
\label{cor:compare-b}
Let $\Pcal$ be a stable partition.
For any $I$, let $J = \pred_{\Pcal}(I)$, and let $I^{(b)} =\pred_{\Pcal}^{(b)}(I), J^{(b)} =\pred_{\Pcal}^{(b+1)}(I)$ for $b\ge 0$.
Then
\begin{enumerate}[label=(\alph*)]
\item $\ed(X[J..I-1], Y[J^{(b)}..I^{(b)}-1]) \le \sing_{X}[I^{(b)},I) + \sing_{Y}[J^{(b)},J)\le \sing[J^{(b)},I)$;
\item $|(I-I^{(b)})-(J-J^{(b)})|=|(I-J)-(I^{(b)}-J^{(b)})|\le \sing_{X}[I^{(b)},I) + \sing_{Y}[J^{(b)},J)\le \sing[J^{(b)},I)$.
\end{enumerate}
\end{corollary}
\begin{proof}
For every $0\le c<b$, by applying \Cref{lem:compare} to $J^{(c)}$ and $I^{(c)}$, we obtain
\[ \ed(X[J^{(c)}..I^{(c)}-1], Y[J^{(c+1)}..I^{(c+1)}-1]) \le \sing_{X}[J^{(c)},I^{(c)}) + \sing_{Y}[J^{(c+1)},I^{(c+1)}).  \]
Then by triangle inequality,
\begin{align*}
\ed(X[J..I-1], Y[J^{(b)}..I^{(b)}-1]) 
&\le \sum_{0\le c< b}\ed(X[J^{(c)}..I^{(c)}-1], Y[J^{(c+1)}..I^{(c+1)}-1])\\
&\le \sum_{0\le c< b}\pbra{\sing_{X}[J^{(c)},I^{(c)}) + \sing_{Y}[J^{(c+1)},I^{(c+1)})}\\
&= \sing_X[I^{(b)},I) +\sing_Y[J^{(b)},J)\\
&\le \sing[J^{(b)},I).
\end{align*}
On the other hand, Item (b) follows from Item (a) by \Cref{fct:length_to_edit}.
\end{proof}

Now we extend it to $b$-stable states.

\begin{corollary}
\label{cor:b}
Let $\Pcal$ be a stable partition.
For $I_1\le I_2$ and any fixed $b\ge0$, let $I_1^{(b)} = \pred_{\Pcal}^{(b)}(I_1),I_2^{(b)} = \pred_{\Pcal}^{(b)}(I_2)$. Then
\[ \abs{(I_1-I_1^{(b)}) - (I_2-I_2^{(b)})} \le b\cdot \sing[I_1^{(b)},I_2).\]
Moreover, the following generalization also holds. Suppose $I_1\le I_2\le \dots \le I_m$, then
\[ \sum_{j=1}^{m-1}\abs{(I_j-I_j^{(b)}) - (I_{j+1}-I_{j+1}^{(b)})} \le b\cdot \sing[I_1^{(b)},I_m).\]
\end{corollary}
\begin{proof}
Since
\[ I_j - I_j^{(b)} = \sum_{c=0}^{b-1} (I_j^{(c)}-I_j^{(c+1)}),\]
we have
\begin{align*}
\abs{(I_j-I_j^{(b)}) - (I_{j+1}-I_{j+1}^{(b)})} 
&\le \sum_{c=0}^{b-1} \abs{(I_j^{(c)}-I_{j}^{(c+1)}) - (I_{j+1}^{(c)}-I_{j+1}^{(c+1)})}\\
&\le \sum_{c=0}^{b-1} \pbra{ \sing_X[I_j^{(c)},I_{j+1}^{(c)}) + \sing_Y[I_j^{(c+1)},I_{j+1}^{(c+1)})}.
\tag{due to \Cref{lem:compare}}
\end{align*}
Summing up over all $j=1,2,\dots,m-1$, we obtain
\begin{align*}
\sum_{j=1}^{m-1}\abs{(I_j-I_j^{(b)}) - (I_{j+1}-I_{j+1}^{(b)})}  
&\le \sum_{j=1}^{m-1}  \sum_{c=0}^{b-1} \pbra{ \sing_X[I_j^{(c)},I_{j+1}^{(c)}) + \sing_Y[I_j^{(c+1)},I_{j+1}^{(c+1)})} \\
&= \sum_{c=0}^{b-1} \pbra{\sing_X[I_1^{(c)},I_m^{(c)}) + \sing_Y[I_1^{(c+1)},I_m^{(c+1)})}\\
&\le b\cdot \sing_X[I_1^{(b-1)},I_m) +b\cdot \sing_Y[I_1^{(b)},I_m^{(1)}) \\
&\le b\cdot \sing[I_1^{(b)}, I_m).  \tag*{\qedhere}
\end{align*} 
\end{proof}

We end this subsection with the following lemma, which will be used in \Cref{sec:stable-and}.
\begin{lemma}
\label{lem:equallength}
Let $\Pcal$ be a stable partition.
For any $I$ and $b\ge 1$, let $J=\pred_{\Pcal}^{(b)}(I)$, and let $d=I-J,I' = I+d$.
Then $\ed(X[I..I'-1],Y[J..I-1]) \le 4b\cdot \sing[J,I')$.
\end{lemma}
\begin{proof}
The function $\pred_\Pcal$ divides the string $Y[J..I-1]$ into $b$ stable segments, which correspond to substrings $\sigma_{-b},\sigma_{-b+1},\dots,\sigma_{-1}$ from left to right, where $\sigma_{-i}:= Y[\pred_{\Pcal}^{(i)}(I).. \pred_{\Pcal}^{(i-1)}(I)-1]$.
Suppose $\Pcal$ induces the substrings $\sigma_0,\sigma_1,\sigma_2,\dots$ immediately after $\sigma_{-1}$ (see \Cref{fig:equallength}).

Assume $I'$ is contained in the stable segment corresponding to substring $\sigma_{q}$ for some $q\ge 0$. We divide the proof into the following two cases.

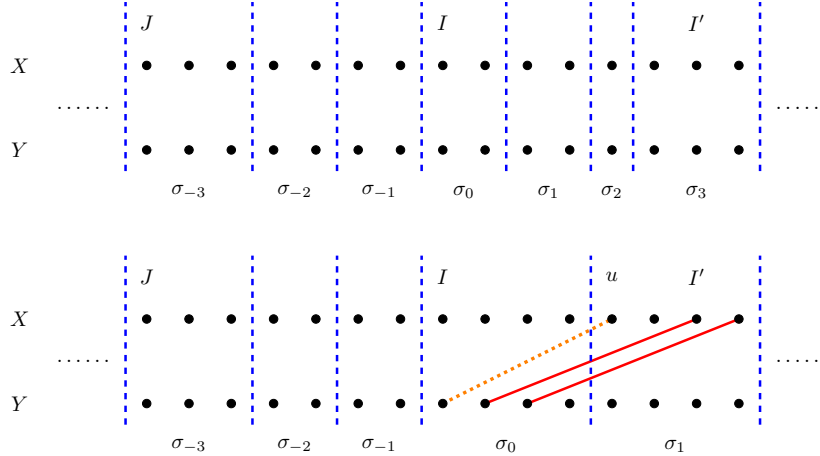
\begin{figure}[ht]
    \centering
    \scalebox{0.8}{\begin{tikzpicture}
    \node at (60pt, 40pt) {$J$};
    \node at (200pt, 40pt) {$I$};            
    \node at (320pt, 40pt) {$I'$};        

    \node at (320pt, -40pt) {$\sigma_{3}$};
    \node at (280pt, -40pt) {$\sigma_{2}$};
    \node at (250pt, -40pt) {$\sigma_{1}$};
    \node at (210pt, -40pt) {$\sigma_{0}$};
    \node at (170pt, -40pt) {$\sigma_{-1}$};
    \node at (130pt, -40pt) {$\sigma_{-2}$};
    \node at (80pt, -40pt) {$\sigma_{-3}$};
    
    \node at (0, 20pt) {$X$};
    \foreach \i in {0,1,2,3,4,5,6,7,8,9,10,11,12,13,14}
        \draw[fill=black] (\i*20 + 60 pt, 20 pt) circle (2pt);

    \node at (30 pt, 0pt) {$\ldots\ldots$};
    \node at (370 pt, 0 pt) {$\ldots\ldots$};

    \node at (0, -20 pt) {$Y$};
    \foreach \i in {0,1,2,3,4,5,6,7,8,9,10,11,12,13,14}
        \draw[fill=black] (\i*20 + 60 pt, -20 pt) circle (2pt);

    \draw[blue,dashed,very thick] (50pt, -30pt) -- (50pt, 50pt);
    \draw[blue,dashed,very thick] (110pt, -30pt) -- (110pt, 50pt);
    \draw[blue,dashed,very thick] (150pt, -30pt) -- (150pt, 50pt);
    \draw[blue,dashed,very thick] (190pt, -30pt) -- (190pt, 50pt);
    \draw[blue,dashed,very thick] (230pt, -30pt) -- (230pt, 50pt);      
    \draw[blue,dashed,very thick] (270pt, -30pt) -- (270pt, 50pt);
    \draw[blue,dashed,very thick] (290pt, -30pt) -- (290pt, 50pt);
    \draw[blue,dashed,very thick] (350pt, -30pt) -- (350pt, 50pt);
    

    \draw[orange,ultra thick,dotted] (200pt, -140pt) -- (280pt, -100pt);   
    \draw[red,very thick] (220pt, -140pt) -- (320pt, -100pt);
    \draw[red,very thick] (240pt, -140pt) -- (340pt, -100pt);   

    \node at (60pt, -80pt) {$J$};
    \node at (200pt, -80pt) {$I$};            
    \node at (280pt, -80pt) {$u$};
    \node at (320pt, -80pt) {$I'$};

    \node at (310pt, -160pt) {$\sigma_{1}$};
    \node at (230pt, -160pt) {$\sigma_{0}$};
    \node at (170pt, -160pt) {$\sigma_{-1}$};
    \node at (130pt, -160pt) {$\sigma_{-2}$};
    \node at (80pt, -160pt) {$\sigma_{-3}$};
    
    \node at (0, -100pt) {$X$};
    \foreach \i in {0,1,2,3,4,5,6,7,8,9,10,11,12,13,14}
        \draw[fill=black] (\i*20 + 60 pt, -100 pt) circle (2pt);

    \node at (30 pt, -120pt) {$\ldots\ldots$};
    \node at (370 pt, -120 pt) {$\ldots\ldots$};

    \node at (0, -140 pt) {$Y$};
    \foreach \i in {0,1,2,3,4,5,6,7,8,9,10,11,12,13,14}
        \draw[fill=black] (\i*20 + 60 pt, -140 pt) circle (2pt);

    \draw[blue,dashed,very thick] (50pt, -150pt) -- (50pt, -70pt);
    \draw[blue,dashed,very thick] (110pt, -150pt) -- (110pt, -70pt);
    \draw[blue,dashed,very thick] (150pt, -150pt) -- (150pt, -70pt);
    \draw[blue,dashed,very thick] (190pt, -150pt) -- (190pt, -70pt);
    \draw[blue,dashed,very thick] (270pt, -150pt) -- (270pt, -70pt);
    \draw[blue,dashed,very thick] (350pt, -150pt) -- (350pt, -70pt);
    
  \end{tikzpicture}}
    \caption{Visualization for the proof of \Cref{lem:equallength}, where the top one is for Case $1$ and the bottom one is for Case $2$. 
    In both cases $b=3$. In Case $1$, $q=3$. In Case $2$, $q=1$.}
    \label{fig:equallength}
\end{figure}

\paragraph*{\fbox{Case 1: $q\ge b$.}}

There are at least $b$ stable segments in $X[I..I'-1]$, which corresponds to substrings $\sigma_0,\sigma_1,\dots,\sigma_{b-1}$. By \Cref{cor:compare-b}, $\ed(\sigma_{i-b},\sigma_{i}) \le \sing[J,I')$ holds for every $0\le i\le b-1$. Then, \[\ed(\sigma_{0}\circ \sigma_{1}\circ \dots \circ \sigma_{b-1},Y[J..I-1])\le \sum_{i=0}^{b-1}\ed(\sigma_{i-b},\sigma_i) \le b\cdot \sing[J,I').\]
By \Cref{prop:edit_dist}, we immediately have 
\[ \ed(X[I..I'-1],Y[J..I-1]) \le 2b\cdot \sing[J,I').\] 

\paragraph*{\fbox{Case 2: $0\le q< b$.}}

Suppose $\sigma_q$ starts at $X[u]$. Then $X[u..I'-1]$ is a prefix of $\sigma_q$. There are $r=(I'-u) - \sing_X[u,I')$ matched edges $(i,j)$ that satisfy $u\le i < I'$ (see the dotted orange line in \Cref{fig:equallength}). 
Since $Y[j]$ must be contained in $\sigma_{q-1}$, we have 
\begin{align*}
\ed(X[u..I'-1], \sigma_{q-1}) 
&\le (I'-u) + |\sigma_{q-1}| - 2r \\
&= 2\cdot \sing_X[u,I') + |\sigma_{q-1}|-(I'-u)\\
&\le 2\cdot \sing[u,I') + |\sigma_{q-1}|-(I'-u).
\end{align*} 
Note that 
\begin{align*}
I'-u  
&= (I'-I) - (|\sigma_0|+|\sigma_1|+\dots + |\sigma_{q-1}|) \\
&=  (|\sigma_{-b}|+|\sigma_{-b+1}|+\dots + |\sigma_{-1}|) - (|\sigma_0|+|\sigma_1|+\dots + |\sigma_{q-1}|) \\
&\ge |\sigma_{-1}| + \sum_{0\le j<q}\pbra{|\sigma_{-b+j}|-|\sigma_{j}|}\\
&\ge |\sigma_{-1}|-q\cdot \sing[J,u).
\tag{due to \Cref{cor:compare-b}}
\end{align*} 
Combining two inequalities, we have
\begin{align*}
\ed(X[u..I'-1],\sigma_{-1}) &\le \ed(X[u..I'-1],\sigma_{q-1}) + \ed(\sigma_{q-1},\sigma_{-1})\\
&\le  2\cdot \sing[u,I')+ q\cdot \sing[J,u) +|\sigma_{q-1}| -|\sigma_{-1}|  +\ed(\sigma_{q-1},\sigma_{-1})\\
&\le 2\cdot \sing[u,I') + q\cdot \sing[J,u)+2\cdot\sing[J,u)
\tag{due to \Cref{cor:compare-b}}\\
&= 2\cdot\sing[J,I') + q\cdot \sing[J,u).
\end{align*}
Hence, by \Cref{prop:edit_dist}, we have
\begin{align*}
\ed(X[I..I'-1],Y[J..I-1]) &\le 2\cdot \ed( X[I..I'-1],\sigma_{-q-1}\circ \cdots \circ \sigma_{-1})\\
&\le 2\cdot \ed(\sigma_0\circ\cdots\circ\sigma_{q-1},\sigma_{-q-1}\circ \cdots \circ  \sigma_{-2}) + 2\cdot \ed(X[u,I'-1],\sigma_{-1})\\
&\le 2q\cdot \sing[J,u) +2\cdot \pbra{2\cdot\sing[J,I') + q\cdot \sing[J,u)}\\
&\le 4b\cdot \sing[J,I').\tag*{\qedhere}
\end{align*}
\end{proof}

\subsection{Catch-up and Stabilize}
\label{sec:stable-and}

This subsection is devoted for \Cref{lem:steps_to_next_stable_state}, which shows a CGK random walk goes from a stable state to a distant stable state with low cost. This process consists of a ``catch-up phase'' (i.e., from a stable state to a distant non-stable state) and then a ``stabilization phase'' (i.e., from a non-stable state to a nearby stable state).

\begin{lemma}[From stable to stable]\label{lem:steps_to_next_stable_state}
Consider a CGK random walk starting from a stable state $(I_0,J_0),I_0>J_0$. Let $D$ be a distance bound satisfying $D\ge I_0-J_0$.  

Consider the first time $T>0$ that either $I_T-J_T>D$, or the following three conditions hold simultaneously: $J_T\ge I_0$, and $I_T\ge 2I_0-J_0$, and $(I_T,J_T)$ is a stable state.\footnote{This time $T$ is almost surely well-defined, as any state $(I,J)$ satisfying $I\ge J>L$ is a stable state.} 
Let $P$ be the number of progress steps before time $T$ and let $S = \sing[J_0,I_T)$.
Then
\[\E\sbra{P - 2000\cdot\pbra{S\cdot D+S^2}} \le 0.\]
\end{lemma}

The proof of this lemma relies on \Cref{lem:equallength} and \Cref{lem:CGK_random_walk} to control the catch-up phase, and the following technical lemma to control the stabilization phase. The proof of \Cref{lem:from_non-stable_to_stable} is deferred to \Cref{sec:stabilize}.

\begin{lemma}[From non-stable to stable]
\label{lem:from_non-stable_to_stable}
Consider a CGK random walk starting from a non-stable state $(\tilde{I}_0,\tilde{J}_0),\tilde{I}_0>\tilde{J}_0$.
Let $\Pcal$ be a stable partition and let $b'$ be such that $\tilde{L}_0<\tilde{J}_0<\tilde{R}_0$ where $\tilde{L}_0=\pred_{\Pcal}^{(b')}(\tilde{I}_0),\tilde{R}_0=\pred_{\Pcal}^{(b'-1)}(\tilde{I}_0)$. 
Let $D$ be a distance bound satisfying $D\ge\tilde{I}_0-\tilde{J}_0$. 

Consider the first time $T'$ that either $(\tilde{I}_{T'},\tilde{J}_{T'})$ is a stable state or $\tilde{I}_{T'}-\tilde{J}_{T'} > D$.\footnote{This time $T'$ is almost surely well-defined, as any state $(I,J)$ satisfying $I\ge J>L$ is a stable state.}  
Let $P'$ be the number of progress steps before time $T'$ and let $S' = \sing[\tilde{L}_0,\tilde{I}_{T'})$. 
Then
\[
\E\sbra{P'-50\cdot \pbra{(\tilde{R}_0-\tilde{J}_0) (\tilde{J}_0-\tilde{L}_0)+S'\cdot D+{S'}^2} } \le 0.
\]
\end{lemma}

Now we present the proof of \Cref{lem:steps_to_next_stable_state}.

\begin{proof}[Proof of \Cref{lem:steps_to_next_stable_state}]
If $S=0$, then $X[J_0..I_T-1]=Y[J_0..I_T-1]$ is periodic string and the period divides $I_0-J_0$; hence $P=0$. Therefore we assume $S\ge1$.

\paragraph*{\fbox{Catch-up phase.}}
Consider the first time $Q$ that either $I_Q-J_Q>D$, or $J_Q\ge I_0$ and $I_Q\ge 2I_0-J_0$. Let $\hat P$ be the number of progress steps before time $Q$. Then 
\begin{equation}\label{eq:catch-up}
\E[\hat P]\le 4\cdot(I_0-J_0)\le 4\cdot D\le 4\cdot\E[S\cdot D].
\end{equation}

If $I_Q-J_Q>D$ or $(I_Q,J_Q)$ is a stable state, then $T=Q$ and $P=\hat P$; hence the bound holds naturally. Thus we focus on the case $J_Q\ge I_0$ and $I_Q\ge 2I_0-J_0$ and $(I_Q,J_Q)$ is non-stable from now on.

\paragraph*{\fbox{Stabilization phase.}}
By the definition of $Q$, we know $T>Q$ is the first time that either $(I_T,J_T)$ is a stable state, or $I_T-J_T>D$. Let $P'$ be the number of progress steps from time $Q$ to time $T$.
Applying \Cref{lem:from_non-stable_to_stable} with $\tilde{I}_0=I_Q,\tilde{J}_0=J_Q$ and inheriting its notations, we have
\begin{equation}\label{eq:stabilize}
\E\sbra{P'-50\cdot\pbra{(\tilde{R}_0-\tilde{J}_0)(\tilde{J}_0-\tilde{L}_0)+S'\cdot D+{S'}^2}}\le0.
\end{equation}
To relate $(\tilde{R}_0-\tilde{J}_0)(\tilde{J}_0-\tilde{L}_0)$ with $S$ and $D$, we will prove the following claim.
\begin{claim}\label{clm:RJJL}
$\E\sbra{\pbra{\tilde{R}_0-\tilde{J}_0}\pbra{\tilde{J}_0-\tilde{L}_0}}\le20\cdot\E\sbra{S\cdot D}$.
\end{claim}

\paragraph*{\fbox{Final bounds.}}
Since $P=\hat P+P'$ and $S'\le S$, by \Cref{eq:catch-up}, \Cref{eq:stabilize}, and \Cref{clm:RJJL} we have
\begin{equation*}
\E\sbra{P-2000\cdot\pbra{S\cdot D+S^2}}\le\E\sbra{P-4S\cdot D-50\cdot\pbra{20S\cdot D+S\cdot D+S^2}}\le0.\tag*{\qedhere}
\end{equation*}
\end{proof}

\begin{proof}[Proof of \Cref{clm:RJJL}]
Assume $(I_0,J_0)$ is $(\Pcal,b)$-stable where $b\ge 1$. By \Cref{lem:predpart} we know $J_0 = \pred_{\Pcal}^{(b)}(I_0)$.
By Item (2) in \Cref{lem:CGK_random_walk}, we have
\begin{align*}
\E\sbra{\abs{(\tilde{I}_0- I_0) - (\tilde{J}_0-J_0)}} 
&\le 4\cdot\ed(X[I_0..2I_0-J_0-1], Y[J_0..I_0-1])\\
&\le 16b\cdot\sing[J_0,2I_0-J_0)
\tag{due to \Cref{lem:equallength}}\\
&\le 16b\cdot S.
\end{align*}
Let $\tilde{J}_0':= \pred_{\Pcal}^{(b)}(\tilde{I}_0)$. By \Cref{cor:b}, we have $\abs{(\tilde{I}_0-\tilde{J}_0') - (I_0-J_0)} \le b \cdot \sing[J_0,\tilde{I}_0)\le b\cdot S$.
Hence 
\begin{equation}\label{eq:temp3}
\E\sbra{\abs{\tilde{J}_0-\tilde{J}_0'}}
\le \E\sbra{\abs{ (\tilde{I}_0-I_0)-(\tilde{J}_0-J_0)}  
+ \abs{ (\tilde{I}_0-\tilde{J}'_0)-(I_0-J_0)}} 
\le 17b\cdot \E\sbra{S}.
\end{equation}
On the other hand, since $\tilde{I}_0\ge\tilde{J}_0,\tilde{J}_0'$, we have
\begin{equation}\label{eq:temp6}
|\tilde{J}_0-\tilde{J}'_0| \le (\tilde{I}_0-\tilde{J}_0) + (\tilde{I}_0-\tilde{J}'_0) \le (\tilde{I}_0-\tilde{J}_0) + (\tilde{I}_0-I_0)+(I_0-J_0)\le 3\cdot D,
\end{equation}

For $\tilde{J}_0$, we know $\tilde{L}_0$ and $\tilde{R}_0$ are the closest boundaries of stable segments induced by $\tilde{I}_0$ and $\Pcal$, hence
\begin{equation}
\min \cbra{\tilde{R}_0 - \tilde{J}_0,\tilde{J}_0 -\tilde{L}_0} \le |\tilde{J}_0-\tilde{J}'_0|. \label{eq:temp1}
\end{equation}
Since $\tilde{R}_0 \ge \tilde{J}_0 \ge I_0$, we apply \Cref{lem:compare} and obtain
$$
\pbra{\tilde{R}_0 -\tilde{L}_0} - \pbra{\pred_{\Pcal}^{(j)}(I_0) - \pred_{\Pcal}^{(j+1)}(I_0)} 
\le \sing[\pred_\Pcal^{(j+1)}(I_0),\tilde{R}_0)
\le S,
$$
for all $0\le j\le b-1$. Taking average over $j$, we obtain
\begin{equation}
\tilde{R}_0 -\tilde{L}_0\le S+\frac{I_0 - J_0}b\le S+\frac Db.\label{eq:temp2}
\end{equation}
Therefore we have
\begin{align*}\label{eq:temp4}
&\phantom{\le}\E\sbra{\pbra{\tilde{R}_0-\tilde{J}_0}\pbra{\tilde{J}_0-\tilde{L}_0}}\\
&\le\E\sbra{\pbra{\tilde{R}_0-\tilde{L}_0}\cdot\min\cbra{\tilde{R}_0-\tilde{J}_0,\tilde{J}_0-\tilde{L}_0}}\\
&\le\E\sbra{\abs{\tilde{J}_0-\tilde{J}'_0}\cdot \pbra{\frac Db+S}}
\tag{due to \Cref{eq:temp1} and \Cref{eq:temp2}}\\
&\le17\cdot\E\sbra{S\cdot D}+3\cdot\E\sbra{S\cdot D}
\tag{due to \Cref{eq:temp3} and \Cref{eq:temp6}}\\
&=20\cdot\E\sbra{S\cdot D}.\tag*{\qedhere}
\end{align*}
\end{proof}

\subsection[Proof of Proposition 3.9]{Proof of \Cref{prop:bounds_on_rho}}
\label{sec:together}

Given previous lemmas to control progress steps, we now prove \Cref{prop:bounds_on_rho}.

\begin{proposition*}[\Cref{prop:bounds_on_rho} restated]
There exists a universal constant $C_4\ge1$ such that the following holds. 
Assume $X,Y$ are two identical length-$L$ strings over alphabet $\Sigma$. 
Assume there exists a size-$M$ matching $(i_1,j_1),\ldots,(i_M,j_M)\in[L]^2$ such that
\begin{itemize}
\item $i_t>j_t$ and $X[i_t]=Y[j_t]$ hold for all $t\in[M]$;
\item $i_1<i_2<\cdots<i_M$ and $j_1<j_2<\cdots<j_M$.
\end{itemize}
Let $\rho=C_4\cdot(L-M)$ and $(\hat I,\hat J)$ be any state satisfying $\hat I-\hat J\ge\rho$.
Then a CGK random walk on $X,Y$ starting from $(\hat I,\hat J)$ will miss $(L,L)$ with probability at least $0.5$.
\end{proposition*}
\begin{proof}
Since $X[1]$ and $Y[L]$ are matched to dummy characters after we extend $X,Y$, there are $K:=\sing[-\infty,+\infty)=2\cdot(L-M-1)$ singletons in total.
Let $d:=\hat I-\hat J$ be the initial distance between the two pointers and let $D:=2\cdot d$.
For a state $(I,J),I\ge J$,
\begin{itemize}
\item if $I=J\le L$ or $I-J>D$, then we say it is a \emph{failure state};
\item if it is not a failure state and $I>L$, then we say it is a \emph{success state}.
\end{itemize}

We stop the CGK random walk when it reaches a success state or a failure state. The former case implies that the random walk misses $(L,L)$. So it suffices to prove that we stop at a success state with probability at least $0.5$.

\paragraph*{\fbox{Phases in the CGK random walk.}}
Let $I_0=\hat I,J_0=\hat J$ and $t_0=0$.
Let $t_1\ge 0$ be the first time that either $(I_{t_1},J_{t_1})$ is a stable state or $I_{t_1}-J_{t_1}>D$.\footnote{By our definition, any state $(I,J)>(L,L)$ is a stable state. Hence $t_1$ is almost surely well-defined.}

For every $i\ge 2$, if $(I_{t_{i-1}},J_{t_{i-1}})$ is neither a success state nor a failure state, we know $J_{t_{i-1}}<I_{t_{i-1}}\le L$ and $I_{t_{i-1}}-J_{t_{i-1}}\le D$. 
Then we recursively define $t_i > t_{i-1}$ to be the first time that either $I_{t_i}-J_{t_i}>D$, or the following three conditions hold simultaneously: $I_{t_i}\ge 2I_{t_{i-1}}-J_{t_{i-1}}$, and $J_{t_i}\ge I_{t_{i-1}}$, and $(I_{t_i},J_{t_i})$ is a stable state. 

Assume we stop at $(I_{t_m},J_{t_m})$, which is either a success state or a failure state.\footnote{Since the random walk walks through $X[1..L],Y[1..L]$ almost surely, $m$ is almost surely well-defined.}
Let $P_i$ be the number of progress steps made during the time interval $[t_i,t_{i+1})$. 
Then $P:= \sum_{i=0}^{m-1}P_i$ is the total number of progress steps before we stop.

\paragraph*{\fbox{Bounds on $\E\sbra{P_0}$.}}
Let $\Pcal$ be an arbitrary stable partition and let $b$ be such that $\pred_{\Pcal}^{(b)}(I_0)\le J_0 < \pred_{\Pcal}^{(b-1)}(I_0)$. Let $L_0=\pred_\Pcal^{(b)}(I_0),R_0=\pred_\Pcal^{(b-1)}(I_0)$.
Since $X[1]$ is matched to $Y[0]$, we know $\pred_\Pcal(1)=0$.
Hence applying \Cref{lem:compare} with $I'=R_0,I=1$, we have 
$$
(R_0-L_0)-(1-0) \le \sing[0,R_0) \le \sing[0,I_0).
$$
Therefore, let $S_0 = \sing[0,I_{t_1})$ and we have
$$ 
(J_0- L_0)(R_0-J_0) 
\le \left\lfloor\frac{R_0-L_0}2\right\rfloor\cdot\left\lceil\frac{R_0-L_0}2\right\rceil\le\pbra{\sing[0,I_0)}^2\le S_0^2.
$$
Thus by \Cref{lem:from_non-stable_to_stable}, we have $\E\sbra{P_0 - 50\cdot\pbra{2\cdot S_0^2 + S_0 \cdot D}}\le 0$.

\paragraph*{\fbox{Bounds on $\E\sbra{P_i},1\le i\le m-1$.}}
Let $S_i:=\sing[J_{t_i},I_{t_{i+1}})$. By  \Cref{lem:steps_to_next_stable_state}, we have 
$$
\E\sbra{P_i - 2000\cdot\pbra{S_i\cdot D+S_i^2}}\le 0.
$$

\paragraph*{\fbox{Final bounds.}}
Note that $\sum_{0\le i<m} S_i \le 2\cdot \sing[0,I_{t_m}) \le 2\cdot K$. This is because $J_{t_{i+1}}\ge I_{t_{i}}$ for all $i\ge 1$, implying each singleton is counted at most twice.
Hence
$$
\E[P] =  \E\sbra{\sum_{i=0}^{m-1}P_i}\le \E\sbra{2000\sum_{i=0}^{m-1}  (S_iD+S_i^2)}\le 2000\cdot\E \sbra{ D\sum_{i=0}^{m-1}S_i+\pbra{\sum_{i=0}^{m-1}S_i}^2 }\le 8000\cdot(K\cdot d+K^2).
$$

For $1\le j< +\infty$, let $r_j$ be the deviation brought by the $j$-th progress step.\footnote{Though we will only use $r_1,\ldots,r_P$, we define it in this way to make the next Cauchy-Schwarz inequality easier to understand.} Then $r_j$ are i.i.d.~random variables with 
$$
\Pr[r_j=0]=1/2,\quad \Pr[r_j=+1]=\Pr[r_j=-1]=1/4.
$$
Hence by Cauchy-Schwarz inequality, we have
$$
\E\sbra{\abs{\sum_{j=1}^P r_j}}
=\E\sbra{\abs{\sum_{j=1}^{+\infty} r_j\cdot1_{\cbra{j\le P}}}}
\le \sqrt{\E\sbra{\pbra{\sum_{j=1}^{+\infty} r_j\cdot1_{\cbra{j\le P}}}^2}} 
= \sqrt{\E\sbra{\frac P2}}\le \sqrt{4000\cdot(K\cdot d+K^2)}.
$$
 
Observe that in the end we have  $I_{t_m}-J_{t_m} = d + \sum_{j=1}^P r_j$. 
By setting $\rho = C_4\cdot(L-M)$ for some large enough constant $C_4$, we have $d\ge\rho\ge C_4\cdot K/2$ and
$$
d\ge4\cdot\sqrt{4000\cdot(K\cdot d+K^2)}
\ge4\cdot\E\sbra{\abs{\sum_{j=1}^P r_j}}.
$$
Then by Markov's inequality, with probability at least $0.5$ we have $|I_{t_m}-J_{t_m}-d| \le d/2$, which indicates $(I_{t_m},J_{t_m})$ is not a failure state. Hence we stop at some success state with probability at least $0.5$.
\end{proof}

\subsection[Proof of Lemma 4.17]{Proof of \Cref{lem:from_non-stable_to_stable}}
\label{sec:stabilize}

This subsection is devoted to the proof of \Cref{lem:from_non-stable_to_stable}. We first prove the following useful lemma.

\begin{lemma}\label{lem:mustpass}
Let $(I,J)$ and $(I',J')$ be two states where $I'\ge I$ and $J'\ge J$. Suppose at least one of the following two conditions holds:
\begin{enumerate}[label=(\arabic*)]
    \item $J\le \pred_{R}^{(b)}(I)$ and $J'\ge \pred_{L}^{(b)}(I')$.
    \item $J\ge \pred_{L}^{(b)}(I)$ and $J'\le \pred_{R}^{(b)}(I')$.
\end{enumerate}
Then every possible CGK random walk from $(I,J)$ to $(I',J')$ must contain a $b$-stable state.
\end{lemma}
\begin{proof}
Let $(I_0,J_0)=(I,J)$ be the starting state.
Then there are four possible transitions during the walk 
$$
(I_{t+1},J_{t+1})\in\cbra{(I_t,J_t), (I_t,J_t+1),(I_t+1,J_t),(I_t+1,J_t+1)}.
$$

We first prove Item (1).
If there is no $b$-stable state on the walk, then by \Cref{lem:b-pred} for all $t$, either $J_t \le \pred^{(b)}_L(I_t)-1$ or $J_t \ge \pred^{(b)}_R(I_{t})+1$. Hence there must exist some $t$ such that $J_t\le \pred^{(b)}_L(I_t)-1$ and $J_{t+1}\ge \pred^{(b)}_R(I_{t+1})+1$. Then
\[ J_{t+1}-1 \ge \pred^{(b)}_R(I_{t+1}) \ge \pred^{(b)}_R(I_{t}) \ge \pred^{(b)}_L(I_{t}) \ge J_t+1. \]
A contradiction.

Now we prove Item (2).
Similarly, if there is no $b$-stable state on the walk, then by \Cref{lem:b-pred} there must exist some $t$ such that $J_t\ge \pred^{(b)}_R(I_t)+1$ and $J_{t+1}\le \pred^{(b)}_L(I_{t+1})-1$. Hence
\[ J_{t}-1 \ge \pred^{(b)}_R(I_t) \ge \pred^{(b)}_L(I_t+1)-1 \ge\pred^{(b)}_L(I_{t+1})-1 \ge J_{t+1}.  \]
A contradiction.
\end{proof}

\begin{corollary}\label{cor:mustpass}
Let $(I,J)$ and $(I',J')$ be two non-stable states where $I'\ge I$ and $J'\ge J$. Let $\Pcal$ be a stable partition and let $b$ be such that $\pred_\Pcal^{(b)}(I)<J<\pred_\Pcal^{(b-1)}(I)$.

If there exists some CGK random walk from $(I,J)$ to $(I',J')$ which does not contain any stable state, then $\pred_\Pcal^{(b)}(I')<J'<\pred_\Pcal^{(b-1)}(I')$.
\end{corollary}
\begin{proof}
By \Cref{cor:b-pred-P}, we have 
$$
\pred_L^{(b)}(I)\le\pred_\Pcal^{(b)}(I)<J<\pred_P^{(b-1)}(I)\le\pred_R^{(b-1)}(I).
$$
As there is no $b$-stable state or $(b-1)$-stable state on the walk, by \Cref{lem:mustpass} and \Cref{cor:b-pred-P} we have
\begin{equation*}
\pred_\Pcal^{(b)}(I')\le\pred_R^{(b)}(I')<J'<\pred_L^{(b-1)}(I')\le\pred_\Pcal^{(b-1)}(I').\tag*{\qedhere}
\end{equation*}
\end{proof}

Now we are ready to prove \Cref{lem:from_non-stable_to_stable}. We restate it here with simpler notations.
\begin{lemma*}[\Cref{lem:from_non-stable_to_stable} restated]
Consider a CGK random walk starting from a non-stable state $(I_0,J_0),I_0>J_0$.
Let $\Pcal$ be a stable partition and let $b$ be such that $L_0< J_0 <R_0$ where $L_0=\pred_{\Pcal}^{(b)}(I_0),R_0=\pred_{\Pcal}^{(b-1)}(I_0)$. 
Let $D$ be a distance bound satisfying $D\ge I_0-J_0$. 

Consider the first time $T$ that either $(I_T,J_T)$ is a stable state or $I_T-J_T > D$. Let $P$ be the number of progress steps before time $T$ and let $S = \sing[L_0,I_T)$. 
Then
\[
\E\sbra{P-50\cdot \pbra{(R_0-J_0) (J_0-L_0)+S\cdot D+S^2} } \le 0.
\]
\end{lemma*}
\begin{proof}
Let $t_0=0$.
Starting from state $(I_{t_0},J_{t_0})$, consider the first time $t_1$ that we reach a state $(I_{t_1},J_{t_1})$ that one of the following three conditions holds: $I_{t_1}-J_{t_1}=I_{t_0}-R_{t_0}$, or $I_{t_1}-J_{t_1}=I_{t_0}-L_{t_0}$, or $I_{t_1}\ge J_{t_1}>L$.\footnote{The third condition implies $(I_{t_1},J_{t_1})$ is a stable state.}
By \Cref{thm:steps_in_random_walk}, the expected number of progress steps is at most $2\cdot(R_{t_0}-J_{t_0})(J_{t_0}-L_{t_0})$.

Note that if $T\le t_1$, then the bounds hold naturally. Therefore, from now on we focus on the case $T>t_1$, which means $I_{t_1}-J_{t_1}\le D$ and we have not reached a stable state during these steps.
Let $L_{t_1}:= \pred_{\Pcal}^{(b)}(I_{t_1})$ and $R_{t_1}:= \pred_{\Pcal}^{(b-1)}(I_{t_1})$. By \Cref{cor:mustpass}, we have $L_{t_1}< J_{t_1}<R_{t_1}$. 
Then similarly we consider the first time $t_2> t_1$ that we reach a state $(I_{t_2},J_{t_2})$ satisfying one of the following three conditions: $I_{t_2}-J_{t_2} = I_{t_1} - R_{t_1}$, or $I_{t_2}-J_{t_2} = I_{t_1} - L_{t_1}$, or $I_{t_2}\ge J_{t_2}>L$, which takes at most
\begin{equation}\label{eq:steps}
2 \cdot(R_{t_1}-J_{t_1})(J_{t_1}-L_{t_1})  
\end{equation}
progress steps in expectation. Now we give an upper bound for \Cref{eq:steps}.
\begin{itemize}
    \item If $I_{t_1}-J_{t_1} = I_{t_0} - R_{t_0}$, then $R_{t_1}-J_{t_1} = (I_{t_0}-R_{t_0}) - (I_{t_1}-R_{t_1})$. Thus
    $$
    \Cref{eq:steps}\le 2\cdot(R_{t_1}-L_{t_1})(R_{t_1}-J_{t_1})
    =2\cdot(R_{t_1}-L_{t_1})\pbra{(I_{t_0}-R_{t_0}) - (I_{t_1}-R_{t_1})}.
    $$
    \item If $I_{t_1}-J_{t_1} = I_{t_0} - L_{t_0}$, then $J_{t_1}-L_{t_1} = (I_{t_1}-L_{t_1}) - (I_{t_0}-L_{t_0})$. Thus
    $$
    \Cref{eq:steps}\le 2\cdot(R_{t_1}-L_{t_1})(J_{t_1}-L_{t_1})
    =2\cdot(R_{t_1}-L_{t_1})\pbra{(I_{t_1}-L_{t_1}) - (I_{t_0}-L_{t_0})}.
    $$
\end{itemize}
In either case, we have
$$
\Cref{eq:steps}\le 2\cdot(R_{t_1}-L_{t_1})\cdot  \max \cbra{(I_{t_0}-R_{t_0}) - (I_{t_1}-R_{t_1}), (I_{t_1}-L_{t_1}) - (I_{t_0}-L_{t_0})}.
$$

Similarly we repeatedly generate $(I_{t_i},J_{t_i})$ as well as $L_{t_i}=\pred_\Pcal^{(b)}(I_{t_i}),R_{t_i}=\pred_\Pcal^{(b-1)}(I_{t_i})$, where we have never reached a stable state, or a state with distance greater than $D$ before and including $(I_{t_m},J_{t_m})$; while starting from $(I_{t_m},J_{t_m})$ we reach $(I_T,J_T)$ which is a stable state, or a state with distance at least $D$. The expected total number of progress steps is
\begin{align}\label{eq:totalsum}
&\phantom{\le}2\cdot(R_{t_0}-J_{t_0})(J_{t_0}-L_{t_0})+2\sum_{i=1}^m(R_{t_i}-J_{t_i}) (J_{t_i}-L_{t_i})\\
&\le2\cdot(R_{t_0}-J_{t_0}) (J_{t_0}-L_{t_0})\nonumber\\
&\phantom{\le}+ 2\sum_{i=1}^m (R_{t_i}-L_{t_i})\cdot  \max \cbra{(I_{{t_{i-1}}}-R_{{t_{i-1}}}) - (I_{t_i}-R_{t_i}), (I_{t_i}-L_{t_i}) - (I_{{t_{i-1}}}-L_{{t_{i-1}}})}.\nonumber
\end{align}

\begin{claim}
\label{clm:claim}
For any possible outcome of $\cbra{(I_{t_i},J_{t_i})}_{i=0}^m$ and $(I_T,J_T)$, let $S'=\sing[L_{t_0},I_{t_m})$ and $S=\sing[L_{t_0},I_T)$, then we have
$$
\Cref{eq:totalsum}
\le 50\cdot\pbra{(R_{t_0}-J_{t_0}) (J_{t_0}-L_{t_0}) + {S'}^2 +  S'\cdot D }
\le 50\cdot\pbra{(R_0-J_0) (J_0-L_0) + S^2 +  S\cdot D }.
$$
\end{claim}

We now show \Cref{clm:claim} implies the lemma. 
Consider a tree with the root node representing the initial state $(I_{t_0},J_{t_0})$. For each node $p$ in the tree, and every possible walk $\lambda$ from state $(I_{t_i},J_{t_i})$ represented by $p$ to a subsequent state $(I_{t_{i+1}},J_{t_{i+1}})$, we draw a directed edge, labeled with $\lambda$, from $p$ to a child node $p'$ representing $(I_{t_{i+1}},J_{t_{i+1}})$. 
We continue this process and include leaf nodes representing stable states or states with distance greater than $D$.
Hence the random process that we are studying is equivalent to a root-to-leaf path on the tree. 

Let $\pi(p,p')$ be the probability of arriving at child node $p'$ if we start from node $p$, and let $\delta(p,p')$ denote the number of progress steps along the walk specified by the edge from $p$ to $p'$. \begin{itemize}
\item For every non-leaf node $p$, let $\sigma(p) = \sum_{\text{$p'$ is a child of $p$}} \pi(p,p')\cdot \delta(p,p')$. We have seen $\sigma(p)\le 2\cdot (J_{t_i}-L_{t_i})(R_{t_i}-J_{t_i})$, where $(I_{t_i},J_{t_i})$ is the state represented by node $p$, and $L_{t_i},R_{t_i}$ are the boundaries of the stable segment containing $J_{t_i}$. 
\item For every leaf node $p$, let $\sigma(p)=-50\cdot\pbra{(R_{0}-J_{0}) (J_{0}-L_{0}) + S^2 + S\cdot D}$, where $S$ is defined for the path from root to $p$.
\end{itemize}

Using this formulation, \Cref{clm:claim} states that for every root-to-leaf path $(p_0,p_1,p_2,\dots,p_{m},p_{m+1})$ where $p_0$ is the root and $p_{m+1}$ is a leaf node, we have $\sum_{i=0}^{m+1}\sigma(p_i) \le 0$. Then,
\begin{align*}
&\phantom{=}\E\sbra{P-50\cdot\pbra{(R_0-J_0)(J_0-L_0)+S\cdot D+S^2}}\\
&=\E_{(p_0,p_1,\dots,p_{m},p_{m+1})}\sbra{ \delta(p_0,p_1)+\delta(p_1,p_2)+\dots + \delta(p_{m},p_{m+1}) +\sigma(p_{m+1})} \\
&=\sum_{(p_0,p_1,\dots,p_{m},p_{m+1})} \pi(p_0,p_1)\pi(p_1,p_2)\cdots \pi(p_{m},p_{m+1})\cdot \pbra{\delta(p_0,p_1)+\delta(p_1,p_2)+\dots + \delta(p_{m},p_{m+1}) +\sigma(p_{m+1})}\\
&=\sigma(p_0)+\sum_{(p_0,p_1)}\pi(p_0,p_1)\sigma(p_1)+\cdots+\sum_{(p_0,p_1,\ldots,p_m)}\pi(p_0,p_1)\pi(p_1,p_2)\cdots\pi(p_{m-1},p_m)\sigma(p_m)\\
&\phantom{=\sigma(p_0)}+\sum_{(p_0,p_1,\ldots,p_m,p_{m+1})}\pi(p_0,p_1)\pi(p_1,p_2)\cdots\pi(p_m,p_{m+1})\sigma(p_{m+1})\\
&=\E_{(p_0,p_1,\dots,p_{m},p_{m+1})} \sbra{\sigma(p_0)+\sigma(p_1)+\dots+\sigma(p_m)+\sigma(p_{m+1})}\\
&\le0.\tag*{\qedhere}
\end{align*}
\end{proof} 

It remains to prove \Cref{clm:claim}, which we divide into two cases.
\begin{proof}[Proof of \Cref{clm:claim}: Case $b=1$]
Note that when $b=1$, we can only have $I_{t_i}-J_{t_i} = I_{t_{i-1}}-L_{t_{i-1}}$, since otherwise $I_{t_i}-J_{t_i}=0$ and hence $(I_{t_i},J_{t_i})$ is a 0-stable state. 
Thus
\begin{align*}
&\phantom{\le}\sum_{i=1}^m (R_{t_i}-J_{t_i}) (J_{t_i}-L_{t_i})\nonumber\\
&\le\sum_{i=1}^m (R_{t_i}-L_{t_i}) (J_{t_i}-L_{t_i})\nonumber\\
&=\sum_{i=1}^m (R_{t_i}-L_{t_i})\cdot \pbra{(I_{t_i}-L_{t_i}) - (I_{t_{i-1}}-L_{t_{i-1}})} \\
&\le\max_{1\le i\le m} \cbra{R_{t_i}-L_{t_i}} \cdot \sum_{i=1}^m\pbra{(I_{t_i}-L_{t_i}) - (I_{t_{i-1}}-L_{t_{i-1}})}\\
&\le\max_{1\le i\le m} \cbra{R_{t_i}-L_{t_i}} \cdot \sing[L_{t_0},I_{t_m}).
\tag{due to \Cref{cor:b}}
\end{align*}
For any $i\in[m]$, by \Cref{lem:compare} we have $(R_{t_i}-L_{t_i}) - (R_{t_0}-L_{t_0}) \le \sing[L_{t_0},R_{t_i})=\sing[L_{t_0},I_{t_i})$.
In addition, since $b=1$ we have $R_{t_i}=I_{t_i}$, which implies 
$$
\max_{1\le i\le m}\cbra{R_{t_i}-L_{t_i}} \le (R_{t_0}-L_{t_0})+\sing[L_{t_0},R_{t_m})  =(I_{t_1}-J_{t_1})+\sing[L_{t_0},I_{t_m})\le D+\sing[L_{t_0},I_{t_m}).
$$
Hence by writing $S'=\sing[I_{t_0},I_{t_m})$, we have
\begin{equation*}
\Cref{eq:totalsum}\le 2\cdot(R_{t_0}-J_{t_0})(J_{t_0}-L_{t_0}) + 2\cdot D\cdot S' +2\cdot{S'}^2. 
\tag*{\qedhere}
\end{equation*}
\end{proof}

\begin{proof}[Proof of \Cref{clm:claim}: Case $b\ge 2$]
We partition $[m]$ into groups where each group contains a contiguous interval of indices. The first group starts from 1. The second group starts from $i=\min\cbra{i\mid L_{t_i}\ge I_{t_1}}$. Similarly, the third group starts from $i'=\min\cbra{i'\mid L_{t_{i'}} \ge I_{t_i}}$, and so on (see \Cref{fig:grouping}).

\begin{figure}[ht]
\centering
\scalebox{0.8}{  \begin{tikzpicture}
    \draw[ultra thick] (0, 0) -- (0, 10pt);
    \node at (0, 20pt) {$L_{t_0}$};

    \draw[thick,red] (20pt, 0) -- (20pt, 10pt);
    \node at (20pt, 20pt) {$J_{t_0}$};

    \draw[ultra thick] (40pt, 0) -- (40pt, 10pt);
    \node at (40pt, 20pt) {$R_{t_0}$};

    \draw[ultra thick] (90pt, 0) -- (90pt, 10pt);
    
    \draw[ultra thick, red] (150pt, 0) -- (150pt, 10pt);
    \node at (150pt, 20pt) {$I_{t_0}$};  
    
    \draw[ultra thick] (0, 0) -- (150pt, 0);

    \draw[ultra thick] (35pt, -35pt) -- (35pt, -25pt);
    \node at (35pt, -15pt) {$L_{t_1}$};

    \draw[thick, red] (75pt, -35pt) -- (75pt, -25pt);
    \node at (72pt, -15pt) {$J_{t_1}$};

    \draw[ultra thick] (85pt, -35pt) -- (85pt, -25pt);
    \node at (88pt, -15pt) {$R_{t_1}$};

    \draw[ultra thick] (135pt, -35pt) -- (135pt, -25pt);

    \draw[ultra thick, red] (185pt, -35pt) -- (185pt, -25pt);
    \node at (175pt, -15pt) {$I_{t_1}$};
    
    \draw[ultra thick] (35pt, -35pt) -- (185pt, -35pt);

    \draw[ultra thick] (105pt, -70pt) -- (105pt, -60pt);
    \node at (105pt, -50pt) {$L_{t_2}$};

    \draw[thick, red] (120pt, -70pt) -- (120pt, -60pt);
    \node at (120pt, -50pt) {$J_{t_2}$};

    \draw[ultra thick] (160pt, -70pt) -- (160pt, -60pt);
    \node at (160pt, -50pt) {$R_{t_2}$};

    \draw[ultra thick] (207.5pt, -70pt) -- (207.5pt, -60pt);

    \draw[ultra thick, red] (270pt, -70pt) -- (270pt, -60pt);
    \node at (270pt, -50pt) {$I_{t_2}$}; 
    
    \draw[ultra thick] (105pt, -70pt) -- (270pt, -70pt);

    
    \draw[ultra thick] (195pt, -105pt) -- (195pt, -95pt);
    \node at (195pt, -85pt) {$L_{t_3}$};

    \draw[thick, red] (212.5pt, -105pt) -- (212.5pt, -95pt);
    \node at (212.5pt, -85pt) {$J_{t_3}$};

    \draw[ultra thick] (232.5pt, -105pt) -- (232.5pt, -95pt);
    \node at (232.5pt, -85pt) {$R_{t_3}$};

    \draw[ultra thick] (282pt, -105pt) -- (282pt, -95pt);
    
    \draw[ultra thick, red] (322.5pt, -105pt) -- (322.5pt, -95pt);
    \node at (335pt, -85pt) {$I_{t_3}$};
    
    \draw[ultra thick] (195pt, -105pt) -- (322.5pt, -105pt);

    \node[scale=1.5] at (282.5pt, -160pt) { $\mathbf{\cdots\cdots} $};
    \draw[ultra thick] (217.5pt, -140pt) -- (217.5pt, -130pt);
    \node at (212pt, -120pt) {$L_{t_4}$};
    \draw[thick, red] (230pt, -140pt) -- (230pt, -130pt);
    \node at (230pt, -120pt) {$J_{t_4}$};
    \draw[ultra thick,red] (357.5pt, -140pt) -- (357.5pt, -130pt);
    \node at (357.5pt, -120pt) {$I_{t_4}$};
    \draw[ultra thick] (277pt, -140pt) -- (277pt, -130pt);
    \node at (277pt, -120pt) {$R_{t_4}$};
    \draw[ultra thick] (317pt, -140pt) -- (317pt, -130pt);
    \draw[ultra thick] (217.5pt, -140pt) -- (357.5pt, -140pt);

    \draw[blue,decorate,decoration={brace,mirror,amplitude=5pt},very thick] (17pt, -22pt) -- (17pt, -75pt) node[midway,xshift=-1.5cm,scale=1.5] {Group $1$};
    \draw[blue,decorate,decoration={brace,mirror,amplitude=5pt},very thick] (177pt, -92pt) -- (177pt, -160pt) node[midway,xshift=-1.5cm,scale=1.5] {Group $2$};

    \draw[blue,dashed,very thick] (185pt, -170pt) -- (185pt, 30pt);
    \draw[blue,dashed,very thick] (322.5pt, -170pt) -- (322.5pt, -40pt);
    
  \end{tikzpicture}}
\caption{Visualization of the grouping in the proof of \Cref{clm:claim}. Here $b=3$.}\label{fig:grouping}
\end{figure}
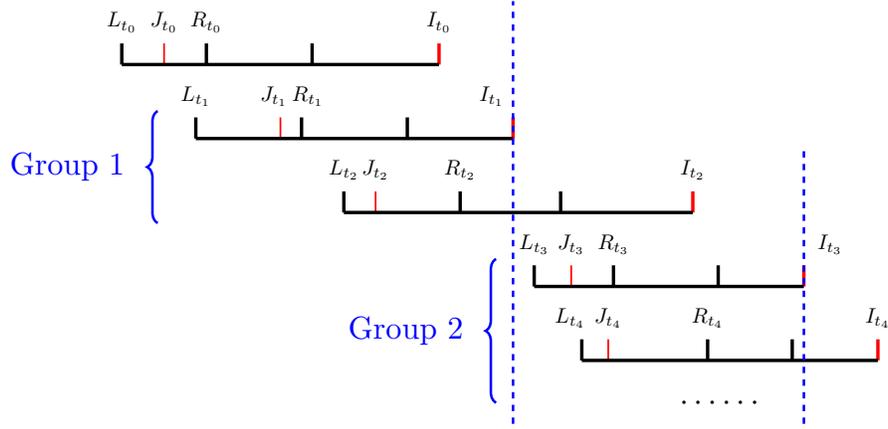

For a particular group $[p..q]$, we have
\begin{align}\label{eq:tobound}
&\phantom{\le}\sum_{i=p}^{q} (R_{t_i}-L_{t_i})\cdot  \max \cbra{(I_{t_{i-1}}-R_{t_{i-1}}) - (I_{t_i}-R_{t_i}), (I_{t_i}-L_{t_i}) - (I_{t_{i-1}}-L_{t_{i-1}})} \nonumber\\
&\le\pbra{\max_{p\le i\le q}\cbra{R_{t_i}-L_{t_i}} }\cdot \sum_{i=p}^{q} \max \cbra{(I_{t_{i-1}}-R_{t_{i-1}}) - (I_{t_i}-R_{t_i}), (I_{t_i}-L_{t_i}) - (I_{t_{i-1}}-L_{t_{i-1}})}.
\end{align}

We now bound the first factor of \Cref{eq:tobound}.
For every $0\le j\le b$, let $I_{t_i}^{(j)} = \pred_{\Pcal}^{(j)}(I_{t_i})$. Hence we have $L_{t_i}=I_{t_i}^{(b)},R_{t_i}=I_{t_i}^{(b-1)}$.
For any $1\le j\le b-1$, by \Cref{lem:compare} we have 
$$
(R_{t_i}-L_{t_i}) - (I^{(j-1)}_{t_i}-I^{(j)}_{t_i})  \le \sing[L_{t_i},I^{(j-1)}_{t_i}).
$$
Taking average over $j$, we obtain
$$
R_{t_i}-L_{t_i} \le \frac{1}{b-1} (I_{t_i}-R_{t_i}) + \sing[L_{t_i},I_{t_i})
\le \frac{1}{b-1} (I_{t_i}-J_{t_i}) + \sing[L_{t_i},I_{t_i})
\le \frac{D}{b-1} + \sing[L_{t_i},I_{t_i}),
$$
which immediately implies
\begin{equation}\label{eq:firstfactor}
    \max_{p\le i\le q}\cbra{R_{t_i}-L_{t_i}} \le \frac{D}{b-1} +  \sing[L_{t_p},I_{t_q}). 
\end{equation}
Hence by \Cref{eq:tobound}, we have
\begin{align}
&\phantom{\le}\sum_{i=1}^{m} (R_{t_i}-L_{t_i})\cdot  \max \cbra{(I_{t_{i-1}}-R_{t_{i-1}}) - (I_{t_i}-R_{t_i}), (I_{t_i}-L_{t_i}) - (I_{t_{i-1}}-L_{t_{i-1}})} \nonumber\\
&\le\frac{D}{b-1}\cdot  \sum_{i=1}^{m} \max \cbra{(I_{t_{i-1}}-R_{t_{i-1}}) - (I_{t_i}-R_{t_i}), (I_{t_i}-L_{t_i}) - (I_{t_{i-1}}-L_{t_{i-1}})}\label{eq:to-bound1} \\
&\phantom{\le}+\sum_{\text{group $g= [p..q]$}} \sing[L_{t_p},I_{t_q})\cdot \sum_{i=p}^{q} \max \cbra{(I_{t_{i-1}}-R_{t_{i-1}}) - (I_{t_i}-R_{t_i}), (I_{t_i}-L_{t_i}) - (I_{t_{i-1}}-L_{t_{i-1}})\label{eq:to-bound2}}.
\end{align}

\paragraph*{\fbox{Bound \Cref{eq:to-bound1}.}}
Applying \Cref{cor:b} and observing $S'=\sing[I_{t_0},I_{t_m})$, we obtain
\begin{align}
&\phantom{\le}\sum_{i=1}^m\max \cbra{(I_{t_{i-1}}-R_{t_{i-1}}) - (I_{t_i}-R_{t_i}), (I_{t_i}-L_{t_i}) - (I_{t_{i-1}}-L_{t_{i-1}})}\nonumber\\
&\le\sum_{i=1}^m \pbra{\abs{ (I_{t_{i-1}}-R_{t_{i-1}}) - (I_{t_i}-R_{t_i})}+\abs{(I_{t_i}-L_{t_i}) - (I_{t_{i-1}}-L_{t_{i-1}})}}\nonumber\\
&\le(b-1)\cdot \sing[R_{t_0},I_{t_m}) + b\cdot \sing[L_{t_0}, I_{t_m})\nonumber \\
&\le(2b-1)\cdot \sing[L_{t_0},I_{t_m})\nonumber\\
&=(2b-1)\cdot S'. \label{eq:first-term-bound}
\end{align}

\paragraph*{\fbox{Bound \Cref{eq:to-bound2}.}}
Note that
\begin{align*}
&\phantom{=}\sum_{i=p+1}^{q} \max   \cbra{(I_{t_{i-1}}-R_{t_{i-1}}) - (I_{t_i}-R_{t_i}), (I_{t_i}-L_{t_i}) - (I_{t_{i-1}}-L_{t_{i-1}})} \\
&=\sum_{i=p+1}^{q} \max\cbra{(R_{t_i} - R_{t_{i-1}}) - (I_{t_i}-I_{t_{i-1}}), (I_{t_i}-I_{t_{i-1}}) - (L_{t_i} -L_{t_{i-1}})} \\
&\le\sum_{i=p+1}^{q} \pbra{ (R_{t_i} - R_{t_{i-1}})  + (I_{t_i}-I_{t_{i-1}})} \\
&=(R_{t_q}-R_{t_p})+(I_{t_q}-I_{t_p})\\
&\le(R_{t_q}-L_{t_q})+(I_{t_p}-R_{t_p}) + (I_{t_q}-R_{t_q})+(R_{t_q}-L_{t_q}).
\tag{due to $I_{t_p}\ge L_{t_q}$}
\end{align*}
We also have
\begin{align*}
&\phantom{=}\max \cbra{(I_{t_{p-1}}-R_{t_{p-1}}) - (I_{t_p}-R_{t_p}), (I_{t_p}-L_{t_p}) - (I_{t_{p-1}}-L_{t_{p-1}})} \\
&\le(I_{t_{p-1}}-R_{t_{p-1}}) + (I_{t_p}-L_{t_p})\\
&=(I_{t_{p-1}}-R_{t_{p-1}}) + (I_{t_p}-R_{t_p}) + (R_{t_p}-L_{t_p}).
\end{align*}
Summing up the above two inequalities and using $I_{t_i}-R_{t_i} \le I_{t_i}-J_{t_i}\le D$, we obtain
\begin{align*}
&\phantom{\le}\sum_{i=p}^{q} \max \cbra{(I_{t_{i-1}}-R_{t_{i-1}}) - (I_{t_i}-R_{t_i}), (I_{t_i}-L_{t_i}) - (I_{t_{i-1}}-L_{t_{i-1}})}\\
&\le4\cdot D + 3\cdot \max\cbra{R_{t_q}-L_{t_q},R_{t_p}-L_{t_p}}\\
&\le4\cdot D + 3\cdot \pbra{\frac{D}{b-1} + \sing[L_{t_p},I_{t_q})}
\tag{due to \Cref{eq:firstfactor}}\\
&\le7\cdot D + 3\cdot \sing[L_{t_p},I_{t_q}).
\end{align*}
Hence \Cref{eq:to-bound2} can be bounded by
\begin{align}
&\phantom{\le}\sum_{\text{group $g= [p..q]$}} \sing[L_{t_p},I_{t_q})\cdot  \pbra{7\cdot D + 3\cdot \sing[L_{t_p},I_{t_q}) }\nonumber\\
&\le7D \cdot \pbra{\sum_{\text{group $g= [p..q]$}} \sing[L_{t_p},I_{t_q})} +3\cdot \pbra{\sum_{\text{group $g= [p..q]$}} \sing[L_{t_p},I_{t_q})}^2\nonumber\\
&\le14D \cdot S' + 12\cdot{S'}^2, \label{second-term-bound}
\end{align}
where the last inequality follows from $\sum_{\text{group $g= [p..q]$}} \sing[L_{t_p},I_{t_q}) \le 2\cdot \sing[L_{t_1},I_{t_m})\le 2\cdot S'$, as our grouping rule ensures that each singleton appears in at most two groups.

\paragraph*{\fbox{Final bounds.}}
Combining \Cref{eq:to-bound1}, \Cref{eq:to-bound2}, \Cref{eq:first-term-bound}, and \Cref{second-term-bound}, we have 
\begin{align*}
\Cref{eq:totalsum}&\le2\cdot(R_{t_0}-J_{t_0}) (J_{t_0}-L_{t_0})  + \frac{2D}{b-1} \cdot (2b-1)\cdot S'+\pbra{28D\cdot S' + 24\cdot{S'}^2}\\
&\le2\cdot (R_{t_0}-J_{t_0}) (J_{t_0}-L_{t_0}) + 34D \cdot S' \cdot  + 24\cdot{S'}^2.\tag*{\qedhere}
\end{align*}
\end{proof}
\section{Discussion}\label{sec:discussions}

Building upon \cite{BelazzouguiZ16}, we present an improved sketching algorithm for edit distance with sketch size $\tilde O(k^3)$. Although the algorithm itself is essentially the same as in \cite{BelazzouguiZ16}, the analysis is more involved. We conclude the paper with a few remarks on further problems.
\begin{itemize}
\item \textbf{Lower bounds.} We conjecture the lower bound for this problem (i.e., $\Qscr_{n,k,\delta}$) is $\tilde\Omega(k^2)$, since $\Theta(k^2)$ is the distortion of the CGK random walk embedding \cite{ChakrabortyGK16}. However, to the best of our knowledge, there is no lower bound beyond $\tilde\Omega(k)$. (Since we do not find any paper formally stating the lower bounds, we present them in \Cref{app:lower_bounds}.)
\item \textbf{Edit distance.} It is natural to wonder if current framework can be pushed further. For example, is it possible that we only run $\tau=O(1)$ rounds of CGK random walks and there will be an optimal matching going through all edges that are common to these walks? Unfortunately this is not true, and we can show $\tau=\Omega(\sqrt k)$ with the following example:  
\begin{align*}
x=Ac_1c_2\cdots c_{k-1} Bc_1c_2\cdots c_{k-1} \underbrace{d\cdots d}_{2k} Ac_1c_2\cdots c_{k-1},\\
y=Bc_1c_2\cdots c_{k-1} \underbrace{d\cdots d}_{2k} Ac_1c_2\cdots c_{k-1} Bc_1c_2\cdots c_{k-1}.
\end{align*}
Then with probability $1-\Theta(1/\sqrt k)$, a CGK random walk walks through $(k,k)$. Note that $\ed(x,y)\le 2\cdot k$ by deleting $x[1..k]$ and inserting $y[4k+1..5k]$. However any edit sequence leaving $(k,k)$ matched will have at least $(2\cdot k+1)$ edits, where the one more edit comes from substituting $x[1]$ with $y[1]$. 
Moreover, this example may generalize to the binary alphabet by replacing each symbol with a short random binary string.
\item \textbf{Ulam distance.} The Ulam distance is the edit distance on two permutations, i.e., $x\in[n]^n$ (resp., $y\in[n]^n$) and $x_i\neq x_j$ (resp., $y_i\neq y_j$) for distinct $i,j$. Our algorithm (as well as the algorithm in \cite{BelazzouguiZ16}) works for Ulam distance with an improved bound $\tilde O(k^{2.5})$. This comes from the following observation: there is no matched edge in the stable zone, hence the length of stable zone is at most $k$, which means we can set $\rho=O(\sqrt L)$ in \Cref{prop:bounds_on_rho}. It would be interesting to improve the algorithm for Ulam distance.
\item \textbf{Only the distance.} Though our algorithm computes edit distance as well as an optimal edit sequence, it is reasonable to relax the problem by simply asking for the distance or even a constant approximation of the distance. However, we are not aware of any result achieving better sketch size in this setting.
\end{itemize}

\section*{Acknowledgements}
We thank Qin Zhang for answering several questions about \cite{BelazzouguiZ16}. C.\ J.\ thanks Virginia Vassilevska Williams for several helpful discussions. We thank anonymous reviewers for their helpful comments.

\bibliographystyle{alphaurl} 
\bibliography{main}

\appendix
\section{Simpler Analysis of CGK}
\label{app:simpler_analysis_of_CGK}

We first restate Item (3) in \Cref{thm:CGK} here and set the number of random walk steps to infinity.

\begin{theorem}[{\cite[Theorem 4.1]{ChakrabortyGK16}}]\label{thm:CGK_app}
Let $\lambda$ be an $\infty$-step random walk on $x,y$, where $p,q$ are the pointers on $x,y$. Then 
$$
\Pr\sbra{\#\text{progress steps in }\lambda\ge\pbra{T\cdot\ed(x,y)}^2}\le O\pbra{\frac1T}.
$$
\end{theorem}
\begin{proof}
Let $z$ be the longest common subsequence of $x,y$. Hence $x$ (and $y$) can be obtained from $z$ by at most $\ed(x,y)$ insertions. We perform a CGK random walk on $z$ with pointer $w$ using the same randomness. Let $\lambda_{x,z}$ (resp., $\lambda_{y,z}$) be view of $\lambda$ on $x,z$ (resp., $y,z$).
Since $\lambda_{x,z}$ and $\lambda_{y,z}$ are projections of $\lambda$, by triangle inequality it suffices to prove
\begin{gather}
\Pr\sbra{\#\text{progress steps in }\lambda_{x,z}\ge\pbra{T\cdot\ed(x,y)}^2}\le O\pbra{\frac1T},\label{eq:CGK_app_1}\\
\Pr\sbra{\#\text{progress steps in }\lambda_{y,z}\ge\pbra{T\cdot\ed(x,y)}^2}\le O\pbra{\frac1T}.\label{eq:CGK_app_2}
\end{gather}

Now we prove \Cref{eq:CGK_app_1}, and the proof of \Cref{eq:CGK_app_2} is analogous. We will find the following classical result useful.

\begin{fact}[e.g.\ {\cite[Theorem 2.17]{levinmarkov}}]\label{thm:random_walk_app}
Let $k$ be some non-negative integer.
Let $\pi$ be a one-dimensional unbiased and self-looped random walk (See \Cref{def:1D_random_walk}) starting from $0$. 
Let $T_0$ be the first time $\pi$ hits $k$. Then
$$
\Pr\sbra{T_0>\pbra{M\cdot k}^2}\le O\pbra{\frac1M}.
$$
\end{fact}

Let $k_{x,z}$ be the number of insertions required to get $x$ from $z$.
Observe that if $p$ is between the $i$-th insertion and the $(i+1)$-th insertion from $z$, we have $p-w\le i$; and when the equality holds the random walk will not have progress steps before $p$ arrives at the $(i+1)$-th insertion. 
Hence, we can safely truncate $\lambda_{x,z}$ at time $t$ for which we have $p_t=w_t+k_{x,z}$. Though $t$ is not necessarily bounded since no progress step occurs outside $x,z$, we can conceptually keep adding progress steps after $p,w$ are outside $x,z$ until $p_t=w_t+k_{x,z}$, which will only increase the count.
Considering the correspondence between progress steps and transitions in a one-dimensional unbiased and self-looped random walk (see \Cref{rmk:progress_step_and_random_walk}), by \Cref{thm:random_walk_app} we have
\begin{align*}
&\Pr\sbra{\#\text{progress steps in }\lambda_{x,z}\ge\pbra{T\cdot\ed(x,y)}^2}\\
\le&\Pr\sbra{\#\text{progress steps in }\lambda_{x,z}\ge\pbra{T\cdot k_{x,z}}^2}\\
\le&\Pr\sbra{\#\text{transitions before }p_t-w_t=k_{x,z}\text{ is at least}\pbra{T\cdot k_{x,z}}^2}\\
\le&~O\pbra{\frac1T}.
\tag*{\qedhere}
\end{align*}
\end{proof}

\section{Lower Bounds}
\label{app:lower_bounds}

In this section we will prove the following lower bounds. Though they are just simple counting arguments and reductions, yet we can't find any paper explicitly stating them. Hence we decide to include the proof here.

\begin{theorem}\label{thm:lower_bound_1}
The sketch size is $\Omega(k\log(n|\Sigma|/k))$ bits if we want to compute edit distance and an optimal edit sequence with probability at least $2/3$.
\end{theorem}
\begin{proof}
Let $\sketchy$ be the sketch of $y$ and $R$ be the maximum number of random bits used.
Let $S$ be the set of triple $(r,x,y)$ where $r\in\bin^R$ is the randomness, $x,y$ are the inputs of length $n$ with $\ed(x,y)\le k$, and we succeed in computing an optimal edit sequence. 
By the assumption, for any fixed $x,y$ there are at least $2/3$ fraction of $r$ such that $(r,x,y)\in S$. Hence 
$$
|S|\ge\underbrace{\frac23\cdot2^R}_\text{number of $r$}\cdot\underbrace{|\Sigma|^n}_\text{number of $x$}\cdot\underbrace{\binom nk\cdot\pbra{|\Sigma|-1}^k}_\text{number of $y$}.
$$

On the other hand, when $(r,x,y)\in S$, we can recover $y$ using $x$ and the edit sequence. Hence we have an injection from $S$ to $\bin^R\times\Sigma^n\times\bin^{|\sketchy|}$, which means
$$
|S|\le2^R\cdot|\Sigma|^n\cdot2^{|\sketchy|}.
$$
By rearranging terms, we have $|\sketchy|=\Omega(k\log(n|\Sigma|/k))$.
\end{proof}

\begin{theorem}\label{thm:lower_bound_2}
Assume $k\le O(\sqrt n),|\Sigma|\ge2$ or $k\le O(n),|\Sigma|\ge 2n$.
The sketch size is $\Omega(k)$ bits if we only want to compute edit distance with probability at least $2/3$.
\end{theorem}
\begin{proof}
We show two reductions from the following theorem.

\begin{theorem}[\cite{HuangSZZ06}]\label{thm:Hamming_lower_bound}
Let $1\le d\le 3N/8$ be a parameter and $K=2d/3$.
Assume Alice gets $X\in\bin^N$ and Bob gets $Y\in\bin^N$, and both $X$ and $Y$ have exactly $K$ ones.
Their goal is to distinguish with probability $2/3$ whether $\Ham(X,Y)\le d$, where $\Ham(\cdot,\cdot)$ is the Hamming distance. Then the number of communication bits is $\Omega(d)$.
\end{theorem}
Now we present the reduction.

\paragraph*{\fbox{Case $k\le O(\sqrt n),|\Sigma|\ge2$.}}

Construct $x$ from $X$ by replacing $0$ with $\underbrace{0..0}_n1\underbrace{0..0}_{2n-1}\underbrace{1..1}_{3n}$ and replacing $1$ with $\underbrace{0..0}_{2n-1}1\underbrace{0..0}_n\underbrace{1..1}_{3n}$. The construction for $y$ is the same. Let $n=6N^2, k=4K$.

Then $\ed(x,y)=2\cdot\Ham(X,Y)\le 4K$. Therefore we can tell if $\Ham(X,Y)\le d$ when we successfully compute $\ed(x,y)$.
Hence by \Cref{thm:Hamming_lower_bound}, the sketch size is $\Omega(d)=\Omega(k)$.
Note that $d\le O(N)$, we have $k\le O(\sqrt n)$. 

\paragraph*{\fbox{Case $k\le O(n),|\Sigma|\ge 2n$.}}

Construct $x$ from $X$ by replacing the $i$-th bit of $X$ with $2\cdot i-X_i\in[2N]$.
The construction for $y$ is the same. Let $n=N,k=2K$.

Then $\ed(x,y)=\Ham(X,Y)\le 2K$. Therefore we can tell if $\Ham(X,Y)\le d$ when we successfully computes $\ed(x,y)$.
Hence by \Cref{thm:Hamming_lower_bound}, the sketch size is $\Omega(d)=\Omega(k)$. Note that $d\le O(N)$, we have $k\le O(n)$.
\end{proof}

We conjecture the conditions in \Cref{thm:lower_bound_2} can be relaxed to $k\le O(n),|\Sigma|\ge 2$ with a better reduction.

\section[Proof of Lemma 2.5]{Proof of \Cref{lem:optimal_matching}}
\label{app:optimal_matching}

\begin{lemma*}[\Cref{lem:optimal_matching} restated]
Let $x,y$ be two strings.
Let $S$ be an optimal edit sequence and $\Mcal(S)$ be its corresponding optimal matching.
\begin{enumerate}[label=(\arabic*)]
\item If $(i,j)\in\Mcal(S)$, then $|i-j|\le\ed(x,y)$.
\item If $u'\le u,v'\le v$ and $u-u'+1=v-v'+1=:L$, then the number of matched edges in $x[u'..u],y[v'..v]$ is at least $L-3\cdot\ed(x,y)-|u-v|$, i.e.,
$$
\abs{\Mcal(S)\cap\big([u'..u]\times[v'..v]\big)}\ge L-3\cdot\ed(x,y)-|u-v|.
$$
\end{enumerate}
\end{lemma*}
\begin{proof}
Let $k:=\ed(x,y)$.
We first prove Item (1). Imagine we start with $x,y$ and perform $S$ to make them equal. Let $(i_t,j_t)$ be the edge $(i,j)$ after performing $t$ edits, then $i_k=j_k$. Note that each edit can change the difference of $i_t,j_t$ by at most $1$, hence $|i-j|=|i_0-j_0|\le k$.

Now we turn to Item (2). Let $U$ be the set of matched edges in $x[u'..u],y[v'..v]$ projected on $y$, then 
$$
U=\cbra{j\in[v'..v]\mid (j,i_j)\text{ is a matched edge for some $i_j$ and }i_j\in[u'..u]}.
$$
Assume without loss of generality $u\ge v$ (or equivalently $u'\ge v'$).
Since $\Mcal(S)$ is non-intersecting, it suffices to prove $|U|\ge L-3\cdot k-(u-v)$.
By Item (1), for any $u'+k\le j\le v-k$, if $(j,i_j)$ is a matched edge, then $i_j\in[u'..u]$. 
On the other hand, there are at most $k$ characters on $y$ that are not covered by $\Mcal(S)$. Hence
\begin{equation*}
|U|\ge\pbra{(v-k)-(u'+k)+1}-k=(u-u'+1)-3\cdot k-(u-v)=L-3\cdot k-(u-v).\tag*{\qedhere}
\end{equation*}
\end{proof}

\section[Proof of Lemma 2.13]{Proof of \Cref{lem:CGK_random_walk}}
\label{app:CGK_random_walk}

We will use the following versions of optional stopping theorem and Borel-Cantelli theorem.

\begin{theorem}[Optional stopping theorem, e.g.\ {\cite[Section 12.5]{GrimmettS20}}]\label{thm:OST}
Let $X=(X_t)_{t\in\Nbb}$ be a discrete-time martingale and $T$ is a stopping time with values in $\Nbb$. If $T$ is almost surely bounded, then $\E\sbra{X_T}=\E\sbra{X_0}$.
\end{theorem}
\begin{theorem}[Borel-Cantelli theorem, e.g.\ {\cite[Section 7.3]{GrimmettS20}}]\label{thm:Borel-Cantelli}
Let $T$ be a non-negative random variable. If
$$
\sum_{i=0}^{+\infty}\Pr\sbra{T>i}<+\infty,
$$
then $T$ is almost surely bounded.
\end{theorem}

Now we prove \Cref{lem:CGK_random_walk}.

\begin{lemma*}[\Cref{lem:CGK_random_walk} restated]
Consider an $\infty$-step CGK random walk $\lambda$ on $x,y$, where $p,q$ are the pointers on $x,y$ respectively. Let $u$ be an index and let $U,V\ge u-1$ be any integers. Then the following hold.
\begin{enumerate}[label=(\arabic*)]
\item Let $T_0$ be the first time that $p_{T_0}\ge u$. Then $\E\sbra{\abs{p_{T_0}-q_{T_0}}}\le 4\cdot\ed(x[1..U],y[1..V])$.
\item Let $T_1$ be the first time that $(p_{T_1}\ge u)\land(q_{T_1}\ge u)$. Then $\E\sbra{\abs{p_{T_1}-q_{T_1}}}\le 4\cdot\ed(x[1..U],y[1..V])$.
\end{enumerate}
\end{lemma*}
\begin{proof}
Let $T$ be the first time that $(p_T\ge u)\lor(q_T\ge u)$.
We first verify $T,T_0,T_1$ are almost surely bounded. 
By Chernoff's bound, we know for any $i\ge6\cdot u$,
$$
\Pr\sbra{T>i},~\Pr\sbra{T_0>i},~\Pr\sbra{T_1>i}\le e^{-\Omega(i)}.
$$
Hence 
$$
\sum_{i=0}^{+\infty}\Pr\sbra{T>i}\le 6\cdot u+\sum_{i=6\cdot u}^{+\infty}e^{-\Omega(i)}<+\infty.
$$
The same calculation holds for $T_0,T_1$. Therefore, by \Cref{thm:Borel-Cantelli} they are almost surely bounded. 

Let $k=\ed(x[1..U],y[1..V])$ and let $z$ be the longest common subsequence of $x[1..U],y[1..V]$. Hence $x[1..U]$ (and $y[1..V]$) can be obtained from $z$ by at most $k$ insertions. In particular, the length of $z$ is at least $U-k\ge u-k-1$.

We perform a CGK random walk on $z$ with pointer $w$ using the same randomness.

\begin{claim}\label{clm:CGK_random_walk_app_1}
$\E\sbra{\abs{p_T-q_T}}\le4\cdot k$.
\end{claim}
\begin{proof}
By \Cref{thm:OST}, $\E\sbra{p_T-w_T}=0$.
Observe that $p_T\le u$ and $p_{T-1}\le u-1$, hence $p_T-w_T\le k$ and
$$
\E\sbra{\abs{p_T-w_T}}=\E\sbra{\abs{p_T-w_T}+(p_T-w_T)}=2\cdot\E\sbra{\max\cbra{p_T-w_T,0}}\le2\cdot k.
$$
Similarly $\E\sbra{\abs{q_T-w_T}}\le2\cdot k$. Hence 
\begin{equation*}
\E\sbra{\abs{p_T-q_T}}\le\E\sbra{\abs{p_T-w_T}}+\E\sbra{\abs{q_T-w_T}}\le4\cdot k.\tag*{\qedhere}
\end{equation*}
\end{proof}

We now prove Item (1). Observe that if $p_T\ge u$, then $T_0=T$.
Otherwise $p_T<u=q_T$ and $T_0>T$, which means $q_{T_0}\ge p_{T_0}$. By \Cref{thm:OST},
$$
\E\sbra{\abs{p_{T_0}-q_{T_0}}\big| p_T,q_T}=\E\sbra{q_{T_0}-p_{T_0}\big| p_T,q_T}=q_T-p_T=\abs{p_T-q_T}.
$$
Hence, $\E\sbra{\abs{p_{T_0}-q_{T_0}}}=\E\sbra{\abs{p_T-q_T}}\le 4\cdot k$.

Let $T_0'$ be the first time that $q_{T_0'}\ge u$. Then by symmetry $|p_{T_0'}-q_{T_0'}|$ shares the same bound as $|p_{T_0}-q_{T_0}|$. Hence Item (2) follows from $T_1=\max\cbra{T_0,T_0'}$.
\end{proof}

\end{document}